\documentclass[11pt]{article}
\usepackage{repeated-games}
\usepackage{physics}
\allowdisplaybreaks
\usepackage{natbib}

\renewcommand{\Pr}{\mathbf{Pr}}

\newcommand{\rl}{r_{\text{learner}}}
\newcommand{\rltilde}{\widetilde{r}_{\text{learner}}}
\newcommand{\ro}{r_{\text{optimizer}}}

\newcommand{\NPtime}{\textsf{NP}}
\newcommand{\Ptime}{\textsf{P}}

\usepackage{cases}

\newcommand{\din}{d_{\text{in}}}
\newcommand{\dout}{d_{\text{out}}}
\newcommand{\Vin}{V_{\text{in}}}
\newcommand{\Vout}{V_{\text{out}}}
\newcommand{\vin}{v_{\text{in}}}
\newcommand{\vout}{v_{\text{out}}}
\newcommand{\uin}{u_{\text{in}}}

\newcommand{\tmax}{t_{\mathrm{max}}}

\newcommand{\Sh}{S_{\mathrm{heavy}}}
\newcommand{\rh}{r_{\mathrm{heavy}}}
\newcommand{\Ti}{T_{\text{init}}}
\newcommand{\Ei}{E_{\text{init}}}
\newcommand{\Ainit}{\mathcal{A}_{\text{init}}}
\newcommand{\Binit}{\mathcal{B}_{\text{init}}}
\newcommand{\Binittilde}{\widetilde{\mathcal{B}}_{\text{init}}}
\newcommand{\dinitin}{d_{\text{init,in}}}
\newcommand{\dinitout}{d_{\text{init,out}}}
\newcommand{\yr}{\bm{y}_{\text{red}}}

\newcommand{\minit}{m_{\text{init}}}
\newcommand{\ninit}{n_{\text{init}}}

\author{
  Angelos Assos\thanks{MIT CSAIL; \texttt{assos@mit.edu}}\\
  \and
  Yuval Dagan\thanks{Tel Aviv University, CS; \texttt{ydagan@tauex.tau.ac.il}}\\
  \and
  Nived Rajaraman\thanks{Berkeley EECS, \texttt{nived.rajaraman@berkeley.edu}}
}

\begin{document}

\title{Computational Intractability of Strategizing against Online Learners}

\date{\today}
\maketitle

\allowdisplaybreaks

\begin{abstract}
Online learning algorithms are widely used in strategic multi-agent settings, including repeated auctions, contract design, and pricing competitions, where agents adapt their strategies over time. A key question in such environments is how an optimizing agent can best respond to a learning agent to improve its own long-term outcomes. While prior work has developed efficient algorithms for the optimizer in special cases — such as structured auction settings or contract design — no general efficient algorithm is known.  

In this paper, we establish a strong computational hardness result: unless $\mathsf{P} = \mathsf{NP}$, no polynomial-time optimizer can compute a near-optimal strategy against a learner using a standard no-regret algorithm, specifically Multiplicative Weights Update (MWU). Our result proves an $\Omega(T)$ hardness bound, significantly strengthening previous work that only showed an additive $\Theta(1)$ impossibility result. Furthermore, while the prior hardness result focused on learners using fictitious play — an algorithm that is not no-regret — we prove intractability for a widely used no-regret learning algorithm. This establishes a fundamental computational barrier to finding optimal strategies in general game-theoretic settings.  
\end{abstract}

\section{Introduction}

Online learning algorithms are often used in scenarios containing other learning agents. This includes repeated auctions where a seller and buyers learn how to construct auctions and to bid, respectively \citep{nekipelov2015econometrics,noti2021bid}; Repeated contracts where a contractor learns how to set up contracts and agents learn how to take actions which maximize their revenue given these contracts \citep{guruganesh2024contracting}; competing retailers that learn how to set prices and to adapt to the demand \citep{du2019pricing,calvano2020artificial}; Information design, where a principle wants to persuade an agent to take actions by deciding which information to share \citep{chen2023persuading}; and many more examples exist.
In all these scenarios, the agents are operating in the same environment and influence each other's utilities.

In this paper we assume that there are two agents, a \emph{learner} and an \emph{optimizer}, and we ask: 
\begin{center}
\emph{Given that the learner is using a well-known learning algorithm to decide on his course of action, what is the best strategy for the optimizer to maximize her own reward (a.k.a. utility)?}
\end{center}
We formalize this setting under the lens of \emph{game theory}: we assume that there exists some \emph{finite normal-form game} between a learner and an optimizer, as explained below. Concretely, this means that the optimizer can take actions from some finite set $\mathcal{A}$, the learner can take actions from some finite set $\mathcal{B}$, and there are functions $A$ and $B$ which determine the reward of the agents: if the optimizer plays action $\bx \in \mathcal{A}$ and the learner plays action $\by \in \mathcal{B}$, the optimizer receives the reward $A(\bx,\by) \in [-1,1]$ and the learner receives the reward $B(\bx,\by) \in [-1,1]$. We assume \emph{repeated game playing}: this means that in every iteration $t=1,\dots,T$, the optimizer chooses an action $\bx(t)$, the learner chooses an action $\by(t)$, and the agents receives reward $A(\bx(t),\by(t))$ and $B(\bx(t),\by(t))$, respectively. The goal of each agent is to maximize their cumulative reward, summed over all the $T$ iterations. For this purpose, the learner is using a standard online learning algorithm which determines which action to take in each iteration, whereas the optimizer's goal is to play according to a strategy that maximizes her cumulative reward given the learner's algorithm.

Perhaps the most fundamental notion by which online learning algorithms are measured is \emph{regret}. That is, the difference between the cumulative reward of the learner and the reward that they would gain if they played a fixed action and if the optimizer played the same sequence of actions, which is defined as
\begin{equation}
    \sum_{t=1}^T B(\bx(t),\by(t)) - \max_{\by\in \mathcal{B}} \sum_{t=1}^T B(\bx(t),\by)~.
\end{equation}
A learning algorithm is no-regret if its regret behaves as $o(T)$ and no-regret is a desirable property. A comprehensive theory has been developed around no-regret learning algorithms, and this includes the fundamental Hedge algorithm, also known as Multiplicative Weights Update (MWU), which is a simple algorithm that attains the asymptotically-optimal regret
\citep{littlestone1994weighted,freund1997decision,arora2012multiplicative}.

A disadvantage of the notion of regret is that it does not take into account scenarios where the other agent (optimizer, in our case) utilizes an adaptive algorithm. In spite of that, standard no-regret learners are still commonly used. This raises the question of how an optimizing agent can behave in environments where other agents use no-regret learning algorithms to make decisions, to maximize its reward. \citet{braverman2018selling} showed that in the context of repeated auctions between one seller and one buyer, if the buyer is using a mean-based no-regret learning algorithm (such as Multiplicative Weights Update) then the seller can extract full welfare, which is more than what they could get if the learner was strategic. This line of work has been generalized in scenarios with one \citep{deng2019prior, rubinstein2024strategizing} or multiple no-regret buyers \citep{cai2023selling}. Beyond mechanism design, similar interactions between a learner and an optimizer were studied in repeated contracts \citep{guruganesh2024contracting} and in repeated information design (Bayesian Persuation) \citep{lin2024persuading}. In all these scenarios, an asymptotically-optimal strategy for the optimizer was developed. However, the solution in each instance was limited to the special game structure at hand. In general games, \cite{deng2019strategizing} show that the general optimization problem is essentially equivalent to solving a high-dimensional optimal control problem, however they do not provide an efficient algorithm. In this paper we show, in fact, that this is no coincidence: there is no efficient algorithm for the optimizer in general games, thereby resolving an open question of \cite{guruganesh2024contracting}.

\subsection{Our Contribution}
Unlike prior work that we discussed above that has provided efficient algorithms only in structured settings or restricted cases, we establish a general computational hardness result. We prove that unless $\mathsf{P} = \mathsf{NP}$, no polynomial-time optimizer can compute a near-optimal strategy against a learner that uses a standard no-regret algorithm, specifically Hedge/MWU. This formally establishes that no general efficient strategy exists for the optimizer in repeated game settings.

We formalize the optimization problem that we study. We assume that the learner's algorithm is Hedge (MWU), and it is parameterized by a step-size $\eta \ge 0$ (cf. \Cref{def:MWU}). Any sequence of optimizer's actions $\bx(1),\dots,\bx(T)$ together with the learning rate $\eta$ determine the learner's actions. This determines the optimizer's reward: $R ( \{ \bx(t) : t \ge 1 \}) = \sum_{t=1}^T A(\bx(t),\by(t))$. The optimizer's objective is stated below.
\begin{task}\label{main-task}
Given $T \in \mathbb{N}$, given truth tables of the reward functions $A,B\colon \mathcal{A}\times\mathcal{B}\to [-1,1]$ and given a step-size $\eta$ for the learner, find a sequence of actions for the optimizer $x(1),\dots,x(T) \in \mathcal{B}$ which approximately maximize its reward $R(\{ \bx(t) : t \ge 1 \})$.
\end{task}

\begin{theorem} \label{theorem:noinit}
Fix absolute constants $\alpha, \beta \in (0,1]$. Let $T$ denote the horizon, suppose the learner plays Hedge with learning rate parameterized as $\eta = 1/T^{1-\alpha}$ (cf. \Cref{def:MWU}) and suppose that the number of strategies per player is bounded by $|\mathcal{A}|, |\mathcal{B}| \le T^\beta$.
If $\Ptime \ne \NPtime$, there exists no algorithm whose runtime is polynomial in $T$ for \Cref{main-task} that outputs $\bx(1),\dots,\bx(T)$ such that:
\begin{align}
    R ( \{ x (t) : t \ge 1 \}) \ge \max_{\{ x^\star (t) : t \in [T] \}} R ( \{ x^\star (t) : t \in [T] \}) - cT + o(T),
\end{align}
for an absolute constant $c > 0$.
\end{theorem}

Our result significantly strengthens the hardness result of \citet{assos2024maximizing}, who show a computational hardness result for attaining the optimal reward in general two-player games. However their result had two significant limitations: (1) their impossibility result is based on approximating the reward of the optimizer up to an additive $\Theta(1)$-factor, whereas ours hold even for approximation up to an additive $\Omega(T)$ term; (2) Their result assumed that the learner uses \emph{follow-the-leader} (a.k.a. best-response, or fictitious play), an algorithm that is not no-regret. 

There are other differences which are not merely conceptual: the structure of the game considered in \citet{assos2024maximizing} has learner action space of size $T$. Therefore establishing any (even $\Omega(T)$) hardness for such instances, does not preclude, for instance, the existence of an efficient algorithm which collects cumulative reward within $\mathcal{O} (\sqrt{|\mathcal{B}|T})$ of that of the best optimizer. Further, it relied on the instability of the learner, which can change actions in each iteration, whereas we consider MWU, which is slow-changing, as is the case of most no-regret learners. 

In contrast, our results show that even when the size of the game is an arbitrarily small (but constant) power of $T$, the regret scaling is $\Omega (T)$, thus ruling out the existence of such algorithms. In a sense, this is the best achievable: if the action space were to be any smaller, say polylogarithmic, the optimizer is now allowed to carry out superpolynomial compute in the size of the game. This essentially rules out approaches which prove hardness based on reducing to $\NPtime$-complete problems of size $\operatorname{poly} (|\mathcal{B}|)$, under $\Ptime \ne \NPtime$. It is conceivable that under stronger assumptions, like ETH, using the same techniques, hardness results may be established even when the size of the game is yet smaller than even a polynomial in $T$.

From a technical point of view, our hardness reduction appeals to a special case of the traveling salesman problem with integral weights, $(1,2)$-TSP, for which constant factor hardness of approximation is known \cite{karpinski2012approximation}. Our result shows $\Omega(T)$ hardness of optimizing against MWU, even when the learning rate is as small as $1/T^{0.99}$, which is the regime where the learner is very slow changing\footnote{$0.99$ can be made any constant arbitrarily close to $1$}. When the learning rate is $0$, the problem becomes easy, since the learner is no longer adaptive.

\subsection{Related work}

\paragraph{Multi-Agent Learning and Strategic Optimization}  
The study of learning in multi-agent environments has a long history in game theory \citep{cesa2006prediction, You04, FL98}. A central question in this field is whether repeated interactions between learning agents lead to convergence to equilibrium \citep{robinson1951iterative, cesa2006prediction}. Classical results establish that when agents repeatedly play against each other using learning algorithms, their average play can converge to various equilibrium concepts. More recent work has shown that when all agents use similar learning algorithms, play can converge significantly faster to equilibria compared to classical results \citep{SALS15, daskalakis2021near, ADF+22, APFS22, FAL+22, PSS22, ZFA+23}. However, a different challenge arises when a strategic optimizer interacts with a learning agent: what is the best strategy for the optimizer to maximize long-term reward?  

\paragraph{Optimizing Against Learning Agents in Structured Settings}  
A series of works have studied how an optimizer should act when interacting with a learning agent in structured settings. In auction design, it has been shown that a seller interacting with a no-regret buyer—whether a single buyer or multiple buyers—can extract nearly the entire welfare, achieving an additional $\Omega(T)$ reward compared to what would be possible against a fully strategic buyer \citep{braverman2018selling, deng2019prior, cai2023selling, rubinstein2024strategizing}. A similar setting was studied in contract design, where a contractor optimally sets contracts for an agent using a learning algorithm \citep{guruganesh2024contracting}, and in Bayesian persuasion, where an information designer persuades a learning agent by strategically revealing information \citep{lin2024persuading}. In each of these cases, efficient algorithms were developed for the optimizer, leveraging the structure of the specific problem.  

\paragraph{General Games: Lack of Efficient Optimizer Algorithms}  
Beyond structured problems, several works have studied optimization against learning agents in general games, where the optimizer’s goal is to maximize its reward in an arbitrary repeated game. \citet{deng2019strategizing} showed how to find an approximately optimal strategy for the optimizer. However, their work only provided a reduction to an optimal control problem and did not provide an efficient algorithm to solve it. Other works have explored similar settings in Bayesian games \citep{mansour2022strategizing}, Pareto-optimal interactions \citep{arunachaleswaran2024pareto}, and optimizers facing arbitrary no-regret learners \citep{brown2024learning}, but none of these provide a general efficient algorithm for the optimizer.  

\paragraph{Computational Hardness in Special Cases}  
While no efficient algorithm is known in general settings, some restricted cases have been solved efficiently. \citet{guo2023optimal} developed a polynomial-time algorithm for an optimizer interacting with an MWU learner in two-player games where each agent has only two actions. \citet{masoumian2024model} provided an efficient algorithm for an optimizer interacting with a learner that has small bounded memory. \citet{CAK23} showed that if a learner best-responds to a known optimizer policy, then the optimizer can achieve more than the per-round Stackelberg value, and they provided efficient algorithms to do so. In first-price auctions, \citet{kumar2024strategically} constructed a bidding algorithm that is both no-regret and robust against a strategic seller.  

\paragraph{Learning the Stackelberg Equilibrium from Repeated Interactions}  
The Stackelberg equilibrium represents the best possible reward an optimizer can obtain in a one-shot game, assuming the learner best-responds to the optimizer’s actions. The problem of computing the Stackelberg equilibrium has been widely studied \citep{CS06, PSTZ19, CAK23}. In repeated interactions, prior work has studied whether an optimizer can learn the Stackelberg equilibrium when it does not know the learner’s reward function in advance. \citet{ananthakrishnan2024knowledge} showed an impossibility result when the learner is strategic and the optimizer has missing information. Other works have developed algorithms for learning a Stackelberg strategy in repeated games under similar uncertainty \citep{BBHP15, MTS12, LGH+22, HLNW22, brown2024learning}. Additionally, \citet{zhang2023steering} explored how a principal can influence online learners to converge to specific equilibria.

\section{Preliminaries} \label{sec:prelim}

Consider a two-player game, where the optimizer's and learner's strategy sets are $\mathcal{A}$ and $\mathcal{B}$, respectively. The players are allowed to play distributions over their action sets (a.k.a. mixed strategies). At time $t$ the optimizer and the learner play $\bx(t) \in \Delta(\mathcal{A})$ and $\by(t) \in \Delta(\mathcal{B})$, respectively. The reward functions for the learner and the optimizer are parameterized by matrices $\bm{A},\bm{B} \in [-1,1]^{|\mathcal{A}| \times |\mathcal{B}|}$. The reward collected by the optimizer and the learner at time $t$ is $\mathbf{x}(t)^{\top} \bm{A} \mathbf{y}(t)$ and $\mathbf{x}(t)^{\top} \bm{B} \mathbf{y}(t)$, respectively. The notation $\mathbf{x}(a,t)$ denotes the probability of action $a \in \mathcal{A}$ at time $t$ under $\mathbf{x}(t)$. Likewise, the notation $\mathbf{y}(b,t)$ is defined. The repeated game is assumed to be played over a horizon of length $T$. Furthermore, for $t \in [T]$, let,
\begin{align}
    \forall b \in \mathcal{B},\quad &\rl(b,t) = \sum\nolimits_{t'=1}^t \mathbf{x}(t')^\top \bm{B} \delta_b
\end{align}
denote the cumulative rewards of the actions of the learner at time $t$.

\medskip

\begin{definition}[Hedge/MWU learner] \label{def:MWU}
    Fix the step size $ \eta \geq 0$ of the algorithm. In a repeated game with matrices $\bm{A}, \bm{B}$ for optimizer and learner, and an initial reward history $r_0 \in \mathbb{R}^m$, plays a mixed strategy, where at time $t$, for any $b \in \mathcal{B}$,
    \begin{equation}
        \by (b,t) = \frac{\exp ( \eta (\rl(b,t) + r_0 (b)) ) }{\sum_{b' \in \mathcal{B}}\exp ( \eta (\rl(b',t) + r_0 (b')) )}
    \end{equation}
    The parameter $\eta$, is referred to as the learning rate.
\end{definition}

\begin{algorithm}
\caption{MWU algorithm}\label{alg:mwu}
\begin{algorithmic}[1]
\Procedure{MWU}{$\bm{B},r_0,T,\eta$}
  \State $Y(0) \gets r_0$ \Comment{Initial reward history}
  \For{$t = 1,2,\dots,T$}
    \State $\by(t) \gets \left( \frac{e^{\eta \cdot Y_1(t)}}{\sum_{i=1}^m e^{\eta \cdot Y_i(t)}}, \dots, \frac{e^{\eta \cdot Y_m(t)}}{\sum_{i=1}^m e^{\eta \cdot Y_i(t)}}\right)$ \Comment{Strategy for the learner at time $t$.}
    \State Learner plays $\by(t)$ and observes $\bx(t)$ \Comment{$\bx(t)$ is the strategy of the optimizer}
    \State $Y(t+1) \gets Y(t) + \bm{B}^{\top}\bx(t)$ \Comment{Update the history}
  \EndFor
\EndProcedure
\end{algorithmic}
\end{algorithm}

\noindent With all the notation in place, we define the main problem we study in this paper.

\begin{definition}[Optimizing against MWU]
For a sequence of pure optimizer strategies $\{ \bx (t) : t \in [T] \}$, we will denote $R_{\eta,r_0} ( \{ \bx (t) : t \ge 1 \})$ as the cumulative reward collected by the optimizer, when playing against multiplicative weights with learning rate and initialization $r_0$. The problem of optimizing against MWU is defined as: find the sequence of pure strategies $\{ \bx (t) : t \in [T] \}$ which maximizes $R_{\eta,r_0} ( \{ \bx (t) : t \ge 1 \})$.
\end{definition}
We notice that the standard definition of MWU is with $r_0=0$ and indeed we will prove hardness for this case. To show this, we will first prove a hardness when $r_0>0$ and provide a reduction between these two scenarios.


\subsection{Reduction Workhorse: $(1,2)$-TSP}

To prove the hardness result (\Cref{theorem:noinit}) we will reduce the problem of optimizing against MWU, to an NP hard problem with an approximation hardness guarantee: we consider a variant of the traveling salesman problem, where the edge weights of the graph are restricted to be either $1$ or $2$. First, we state the standard ``minimization'' version of $(1,2)$-TSP.

\begin{definition}[$(1,2)$-TSP]
The $(1,2)$-TSP problem is specified by a weighted complete graph $G = (V,W)$ on $n$ vertices where $W \in \{ 1,2 \}^{\binom{V}{2}}$. The weight of a Hamiltonian cycle (i.e., each vertex is visited exactly once) $v_1 \to v_2 \to \cdots \to v_n \to v_1$ is $\sum_{i=1}^{n} W (v_i, v_{i+1})$. The objective is to find a Hamiltonian cycle in $V$ having minimum weight.
\end{definition}

\noindent $(1,2)$-TSP admits constant-factor hardness of approximation, as illustrated by the following result.

\begin{theorem}[Hardness of $(1,2)$-TSP \citep{karpinski2012approximation}] \label{theorem:TSP-hardness}
Unless $\Ptime = \NPtime$, there is no polynomial time approximation algorithm for $(1,2)$-TSP achieving a multiplicative approximation factor of $535/534 - \epsilon$ for any $\epsilon > 0$.
\end{theorem}

\noindent It is worth pointing out that $(1,2)$-TSP also admits an additive hardness guarantee. The weight of the optimal Hamiltonian cycle must be between $n$ and $2n$; therefore, $(1+c)$-factor multiplicative approximation hardness for this problem translates to additive hardness of $cn$. This guarantee is what we will use to establish the additive hardness in \Cref{theorem:noinit}. Prior to discussing the reduction, we will refactor $(1,2)$-TSP as a maximization problem, which inherits similar hardness guarantees.
\medskip

\begin{definition}[$(1,2)$-maxTSP]
An instance of $(1,2)$-TSP specified by $G=(V,W)$ is in one-one correspondence with an instance of $(1,2)$-maxTSP on $\widetilde{G} =(V,\widetilde{W})$ by the relation $\widetilde{W}(e) = 3 - W(e)$. This corresponds to replacing all the weight-$2$ edges by weight-$1$ edges, and vice versa. The objective is to find an Hamiltonian cycle in $\widetilde{G}$ with maximum weight.
\end{definition}

The following theorem shows constant approximation factor hardness for $(1,2)$-maxTSP (this follows trivially from \cite{karpinski2012approximation}).

\begin{theorem} \label{theorem:maxTSP-hardness}
Unless $\Ptime = \NPtime$, there is no polynomial time approximation algorithm for $(1,2)$-maxTSP achieving a multiplicative approximation factor of $1067/1068 + \epsilon$ for any $\epsilon > 0$.
\end{theorem}

\begin{proof}
    Consider any instance $G = (V,W)$ of $(1,2)$-TSP and the associated instance $\widetilde{G} = (V,\widetilde{W})$ of $(1,2)$-maxTSP. Consider any Hamiltonian cycle $C$ on $V$. By the relation between $W$ and $\widetilde{W}$, we have that $W(C) = 3n - \widetilde{W}(C)$. Taking the minimum over Hamiltonian cycles $C$, for any cycle $C$,
    \begin{equation*}
        \frac{3n - W(C)}{3n - \min_{C'} W(C')} = \frac{\widetilde{W}(C)}{\max_{C'} \widetilde{W}(C)}
    \end{equation*}
    Unless $\Ptime = \NPtime$, it is not possible to find in polynomial time a Hamiltonian cycle $C$ such that $W (C) / \min_{C'} W (C) \le 535/534-\epsilon$. By contradiction this implies that we cannot find a Hamiltonian cycle $C$ such that $\widetilde{W} (C) / \max_{C'} \widetilde{W} (C) \ge \frac{1067}{1068}+\varepsilon$ in polynomial time, unless $\Ptime = \NPtime$.
\end{proof}

\section{Lower bound against MWU with non-zero reward history}

The main result we show in this section is an $\Omega(T)$ hardness result for optimizing against MWU when initialized with an arbitrary reward history. In the subsequent section, we will extend this result to optimizing against MWU in the absence of any initialization.

\begin{theorem} \label{theorem:with_initialization}
    Fix arbitrary constants $\alpha, \beta \in (0,1]$. Suppose the learner plays MWU (\Cref{def:MWU}) with a specific non-zero initialization of reward history $r_0$, and learning rate parameterized as $\eta = 1/T^{1-\alpha}$. There exist game matrices $\bm{A}, \bm{B} \in [-1,1]^{|\mathcal{A}| \times |\mathcal{B}|}$, where $|\mathcal{A}|, |\mathcal{B}| \le T^\beta$, such that unless $\Ptime = \NPtime$, there exists no polynomial time optimizer which finds a sequence of pure strategies for the optimizer $\{ \bx (t) : t \in [T] \}$ satisfying,
\begin{align}
    R_{\eta,r_0} ( \{ \bx (t) : t \in [T] \}) \ge \max_{\{ \bx^\star (t) : t \in [T] \}} R_{\eta,r_0} ( \{ \bx^\star (t) : t \in [T] \}) - cT + o(T)
\end{align}
For a sufficiently small absolute constant $c > 0$ which depends on $\alpha$ and $\beta$.
\end{theorem}

\noindent As hinted at previously, the proof of this result will go through via a reduction to $(1,2)$-maxTSP. We discuss the construction next.

\subsection{Structure of the reduction} \label{sec:structure}

Define $E = (V \times V) \setminus \{ (v,v) : v \in V \}$, the set of directed edges of the complete digraph on $V$. Given an instance of $(1,2)$-maxTSP on the vertex set $V$, specified by a set of edge weights $\{ W_e : e \in E \}$, define the following game structure for both agents:
\begin{subequations}
\begin{align}
    &\textbf{Optimizer's action space } \Ainit = E \textbf{ of size } \minit = |V|^2 - |V| \\
    &\textbf{Learner's action space } \Binit = E \cup \Vin \cup \Vout \textbf{ of size } \ninit = |V|^2 + |V|
\end{align}
\end{subequations}
where $\Vin$ and $\Vout$ are auxiliary copies of $V$. For any vertex $v \in V$, we will use $\vin$ and $\vout$ to denote the corresponding vertex in $\Vin$ and $\Vout$, and vice versa. Recall that we parameterize MWU's learning rate as $\eta = 1/T^{1-\alpha}$ for some $\alpha \in (0,1]$. We will later choose $V$ to be of size $n = T^{\min \{\alpha,\beta\}/3}$. Thus, the learner's and optimizer's action spaces are of size $\mathcal{O} (T^{2\beta/3})$.

The reduction's idea is to create a game instance, where the optimal strategy for the optimizer is to play along some approximate solution of the maxTSP problem. Specifically, some edge $e^\dagger$ is forced as the first edge to play, and the optimizer has to proceed with an approximate solution to the TSP problem, playing each edge along this path for $k$ rounds. The reward matrix will be constructed such that, when the optimizer plays along such strategy, it receives a utility that is proportional to its path weight, hence it would like to find a maximal-weight path. Notice that the optimizer does not have to play according to some path. To ensure that it cannot gain more utility deviating from playing each edge in a path for $k$ rounds, the historical rewards are defined such that if the learner significantly diverges from such a play-structure, the learner will play actions in $\Vin$ and $\Vout$ from that point onwards, and these yield no utility for the optimizer.

\noindent The game matrices are defined as follows:
\begin{subnumcases}{\text{For } (w,x) \in E,\quad \bm{A} ( \underbrace{(u,v)}_{\text{opt.}}, \underbrace{(w,x)}_{\text{lr.}} ) =}
    W_{(w,x)} &if $(u,v) = (w,x)$ \\
    W_{(w,x)} &if $u=x$, \\
    0 &\text{otherwise.}
\end{subnumcases}
\begin{align} \label{eq:blacklozenge}
    \text{For } a \in \Ainit,\ b \in \Binit \setminus E, \quad \bm{A} ( a, b ) &= 0.
\end{align}
Thus, if the learner plays an edge $e$, the optimizer collects $W_e$ units of reward for either responding with the same edge or responding with one which stems from the out-vertex of this edge. The optimizer surely collects no reward if the learner had chosen an action in $\Vin$ or $\Vout$. On the other hand, for some $\varepsilon > 0$ and an arbitrary edge $e^\dagger = (u^\dagger, v^\dagger) \in E$, the game matrix for the learner is,
\begin{subnumcases}{
\text{For } (w,x) \in E,\quad \bm{B} ( \underbrace{(u,v)}_{\text{opt.}}, \underbrace{(w,x)}_{\text{lr.}} ) = }
         0 \qquad &if $(u,v) = e^\dagger$ \label{eq:20a}\\
         1 \qquad &if $(u,v) \ne e^\dagger$ and $(w,x) = (u,v)$ \label{eq:20b}\\
         -\varepsilon &if $x=u,$ \label{eq:20c}\\
         0 &otherwise. \label{eq:20d}
\end{subnumcases}
\begin{align}
    \text{For } b \in \Vin,\quad \bm{B} ( (u,v), b) &= 2 \mathbbm{1} ( b = v)\\
    \text{For } b \in \Vout,\quad \bm{B} ( (u,v), b ) &= 2 \mathbbm{1} ( b = u)
\end{align}
Thus, the cumulative reward of an edge for the learner is the number of times the same edge was played by the optimizer minus the number of times the optimizer chooses an edge emanating from its out vertex scaled by $\varepsilon$.

\medskip

\noindent Recall that in the statement of \Cref{theorem:with_initialization} we are allowed to assume that the cumulative rewards of actions played by MWU are initially non-zero. This ``initial reward'' influences the probability of actions chosen by MWU. In particular, for $k > 0$ to be determined later, suppose the cumulative rewards are initialized as follows,
\begin{align}
    \text{For } \vin \in \Vin,\ r_0 (\vin) &= -k \label{init:1}\\
    \text{For } \vout \in \Vout,\ r_0 (\vout) &= -k \label{init:2}\\
    \text{For } (u,v) \in E \setminus \{ e^\dagger \},\ r_0 ( (u,v) ) &= 0 \label{init:3}\\
    r_0 ( e^\dagger ) &= k \label{init:4}
\end{align}
In particular, when $k \gg \log (n)/\eta$, the initial distribution determined by MWU places most of its mass on the edge $e^\dagger$. As we will see later, the presence of the edge $e^\dagger$ is vital: it will ensure that there always exists an action with high cumulative reward ($\approx k$) for the learner. On the other hand, the vertices $\vin \in \Vin$ and $\vout \in \Vout$ will serve to keep a handle on the number of edges the optimizer may play sinking / emanating out of a single vertex. These vertices begin with a high negative reward, $-k$, but gain reward at twice the rate of edges in $\mathcal{B}_{\text{init}}$. Once any of these vertices gain too much reward, MWU will displace most of its mass toward such actions, preventing the optimizer from gaining significant reward then onward.

\subsection{Proof of \Cref{theorem:with_initialization}: non-zero historical rewards}

We start by arguing that the optimizer should not play an edge significantly more than $k$ times, otherwise it will stop gaining reward. Define $V_{\text{excess}}(\Delta,t)$ as the set of vertices $v$ for which there exists an incident edge that has been played more than $k+\Delta$ times by the optimizer by time $t$, where $\Delta$ is small compared to $k$.

\begin{definition} \label{def:excess}
For $v\in V$ and $t=1,\dots,T$ denote by $\din(v,t)$ (resp. $\dout(v,t)$) the number of times that the optimizer played an edge incoming (resp. outgoing) $v$, in iterations $1,\dots,t$. We refer to $\din(v,t)$ and $\dout(v,t)$ as \emph{in-degree} and \emph{out-degree} of $v$ throughout the play.
Let $V_{\text{excess}} (\Delta,t)$ denote the set of vertices which have in-degree or out-degree exceeding $k + \Delta$. Namely,
\begin{equation}
    V_{\text{excess}} (\Delta,t) = \left\{ v \in V : \max \{ \din (v,t), \dout (v,t) \} \ge k + \Delta \right\}.
\end{equation}
\end{definition}

\noindent In the following lemma we prove that the optimizer will be punished for playing edges that emanate or sink in a vertex $v$ too many times. Specifically, if at time $\tmax$, $V_{\text{excess}} (\Delta,t) $ is non empty, then from that point onwards the rewards the optimizer can obtain are limited. The idea is that if the optimizer exceeds the threshold of $k + \Delta$ for a vertex $v$, one of the actions $\vin, \vout$ will receive much higher reward than the other actions of the optimizer. The MWU learner will then place almost all of its probability mass on this action, for which the optimizer earns no reward.

\begin{lemma} \label{lemma:excess}
Suppose $\Delta \ge \frac{1}{\eta} \log \left(2 n^2 T \right)$. If at any time $t >0$, $V_{\text{excess}} (\Delta,t)$ is nonempty, then, the optimizer can collect at most $\frac{1}{T}$ reward at time $t+1$.
\end{lemma}

We now use the above lemma to prove the following corollary. Denoting the first time $V_{\text{excess}}(\Delta, t)$ becomes non-empty as $\tmax$, then across timesteps $\tmax + 1, \tmax + 2, \dots, T$, the optimizer can collect cumulative reward at most $\frac{1}{T} \times T = 1$.

\begin{corollary} \label{corr:excess}
    Suppose $\Delta \ge \frac{1}{\eta} \log \left(2 n^2 T \right)$. Let the first time $V_{\text{excess}}(\Delta, t)$ becomes non-empty be denoted $\tmax$ (if it exists), and suppose the optimizer has collected reward $R(\tmax)$ so far. Then, by the end of the repeated interaction the optimizer can collect cumulative reward at most $R(\tmax) + 1$.
\end{corollary}

Next, we argue that the optimizer will not gain much reward by playing edges which emanate from vertices such that all incoming edges to the vertex are played with low frequency by the optimizer. Formally, define $\Sh (\Delta, t)$, the set of \textit{heavy vertices} at time $t$: this is the union of $\{ v^\dagger \}$ with the vertices $v \in V$ which satisfy the condition that there exists an edge of the form $e = (\cdot,v)$ such that  $E (e, t) \ge k (1 - \varepsilon ) - \Delta$. Since we restrict ourselves to time $\tmax$, vertices have in-degree at most $k + \Delta$ (as induced by the edges played by the optimizer). Thus, the heavy vertices are those with high in-degree, induced almost entirely by $\approx k$ copies of a single incoming edge.
\begin{definition} \label{def:heavy}
        Define $\Sh(\Delta, t)$ to be the set of vertices $v \in V$ for which there exists an edge $(\cdot, v)$ for which this edge has been used by the optimizer at least $k(1-\varepsilon) - \Delta$ times plus $v^\dagger$. Precisely:
    \begin{align}
        \Sh(\Delta, t) = \{v^\dagger\} \cup \{v \in V | \exists u \in V, Q((u,v), t) \geq k(1-\varepsilon) - \Delta\}
    \end{align}
    where $Q((u,v),t)$ is the number of time that $(u,v)$ was played by the optimizer in iterations $1,\dots,t$.
\end{definition}

\noindent In \Cref{lemma:heavy}, we argue that if we consider the set of heavy vertices, $\Sh (\Delta,\tmax)$, the optimizer essentially only collects reward for the edges played which emanate from a vertex in this set. We will denote this cumulative reward by $\rh (\Delta, \tmax)$. The idea is that any vertex $v \not\in \Sh (\Delta,\tmax)$ will never have a single edge played more than $k(1-\varepsilon) - \Delta$ times incoming into it. The optimizer may only collect reward for an edge emanating from this vertex if:
    \begin{enumerate}
        \item MWU places high mass on some edge incoming into this vertex. This is not possible because every edge incoming into $v$ has low frequency (as induced by the optimizer).
        \item MWU places high mass on the same edge. This is not possible until the same edge has been pulled many (at least approximately $k (1-\varepsilon)$) times by the optimizer. After this point, the same edge can be played at most until its frequency hits $k + \Delta$. In the small interim interval is when any reward can be collected for picking this edge.
    \end{enumerate}

\begin{lemma} \label{lemma:heavy}
Let $\rh (\Delta, t)$ count the contribution of the total reward collected by the optimizer at time $t$ arising only from edges which emanate from a vertex in $\Sh (\tmax)$. Then,
\begin{align}
    \ro (\tmax) \le \rh ( \Delta, \tmax ) + 2n (k \varepsilon + 2\Delta) + 3
\end{align}
\end{lemma}

\noindent The above assertions have a few consequences. First, we have that the optimizer should be very careful not to exceed playing one action more than $k + \Delta$ times, and risks receiving almost no reward for the rest of the game (\Cref{corr:excess}). We then define the 'heavy' vertices to be the vertices for which an edge incident to them have been played a significant amount of times. Then, with \Cref{lemma:heavy}, we argued that it is enough to look at the reward of the edges emanating from 'heavy' vertices of the optimizer, $\rh(\Delta, \tmax)$, as they make up most of the reward of the optimizer.
\smallskip

We will now try to unpack and upper bound $\rh(\Delta, \tmax)$. Notice that for every $v \in \Sh(\Delta, \tmax)$, there is a unique vertex $u$ for which the action $(u,v)$ is played more than $k(1-\varepsilon) - \Delta$ times. There cannot be more than one, otherwise we would have $\din(v, t) = 2(k(1-\varepsilon) - \Delta) > k + \Delta$, which is a contradiction to the fact that $V_{\text{excess}}$ is empty prior to reaching time $\tmax$. Similarly, one can prove that for every vertex in $v$ there can be at most one action $(v,w)$ having reward at least $k(1-\varepsilon) - \Delta$. This motivates the following graph construction of $G'$.

\begin{definition} \label{def:auxiliary}
    Construct a graph $G'$ using the vertices of $V$ as follows; add an edge from $u \to v$ if $v \in \Sh(\Delta, \tmax)$ and $E((u,v), \tmax) \geq k(1-\varepsilon) - \Delta$.
\end{definition}

The first observation, as already hinted, is that $G'$ is a graph in which all vertices have in-degree and out-degree at most one. That means that we may split $G'$ into disjoint paths and cycles, which we denote as $C_1, C_2, \cdots C_s$. Our goal will now be to bound the rewards obtained by actions emanating from the paths (\Cref{lemma:path}) and cycles (\Cref{lemma:cycle}) of the Graph $G'$ and combining the two to get a global upper bound of the reward collected by optimizer edges emanating from vertices in $\Sh(\Delta, t)$  (\Cref{lemma:combine}). We begin by bounding the rewards of  actions/edges that emanate from vertices of a specific path of the graph $G'$; We argue that the total reward collected by the optimizer for edges emanating from these vertices is approximately upper bounded by the sum of the weights (in the maxTSP instance) of the edges along this path along with a free edge pointing from the last vertex to an arbitrary vertex, multiplied by $k$. Without loss of generality, we assume that the component $C_1$ contains the edge $e^\dagger$.

\begin{lemma} \label{lemma:path}
Consider any \emph{path} $C_i$ in $G'$ composed of vertices $z_1,z_2,\cdots,z_m \in V$. The total reward collected by the optimizer for edges which stem from vertices on $C_i$ is upper bounded by,
\begin{align}
    \sum_{i=1}^{m} (k (1 + \varepsilon) + 3\Delta) W (z_i, z_{i+1}) +  2 (k + \Delta) \mathbb{I} (i=1).
\end{align}
where $z_{i+1}$ is an arbitrary (free) vertex.
\end{lemma}

\begin{proof}
    Since $C_i$ is assumed to be a path, the vertex $z_i$ has no incident edge in $G'$ . We will prove that the following quantity is an upper bound on the total reward the optimizer collects via playing edges emanating from vertices in $C_i$,
    \begin{align} \label{eq:path-ub-1}
         \underbrace{\sum_{i = 2}^{m-1} \Big\{ \alpha_i W( z_i, z_{i+1}) + 2 (k + \Delta - \alpha_i) \Big\}}_{(i)} + \underbrace{2  \sum_{v \in V} W((z_m, v),\tmax)}_{(ii)} + \underbrace{2 (k+\Delta) \mathbb{I} (i=1)}_{(iii)}
    \end{align}
    where $\alpha_i \ge k (1 - \varepsilon) - \Delta$ for all $i$. We will later simplify this bound to prove the lemma. The first term $(i)$ can be attributed to the fact that for every vertex in the path $z_{i+1}$ for $i \ge 1$, there is a high number of edges from the previous vertex in the path, $z_i$. The number of such edges is captured by $\alpha_i \ge k (1 - \varepsilon) - \Delta$. This vertex $z_{i+1}$ may have $(k + \Delta) - \alpha_i$ remaining edges emanating out from it, and these edges may collect the maximum possible reward of $2$ (since each edge weight is in $\{ 1,2 \}$). Thus the total reward for edges coming out of $z_{i+1}$ is $\alpha_i W(z_i, z_{i+1}) + 2(k+\Delta - \alpha_i)$ The term $(ii)$ accounts for the reward the last vertex in the path, $z_m$, can collect from edges emanating out from it. Since $z_m \in \Sh(\Delta, \tmax)$ every edge emanating from it can collect a reward of at most $2$. Term $(iii)$ accounts for the fact that specifically for $C_1$, the vertex $v^\dagger$ is always in $\Sh$ by definition. The reason for this choice is that even if there are no in-edges to $v^\dagger$, MWU associates high initial reward to the edge $(u^\dagger,v^\dagger)$ (so this vertex ``behaves'' as though the optimizer had chosen many in-edges incident on it). Simplifying eq.~\eqref{eq:path-ub-1} further, assuming the in-degree bound of $k + \Delta$ on vertices, we arrive at the following upper bound on the total reward collected by the optimizer from playing edges emanating from $\{ z_1,\cdots,z_m \}$,
    \begin{align}
        &\sum_{i=2}^{m-1} \underbrace{\Big\{ (k ( 1 - \varepsilon) - \Delta) W (z_i, z_{i+1}) + 2 (k \varepsilon + 2 \Delta) \Big\}}_{(i)} + \underbrace{2 (k + \Delta)}_{(ii)} + \underbrace{2(k+\Delta) \mathbb{I} (i=1)}_{(iii)} \label{eq:path-ub-2}\\
        &\le \sum_{i=2}^{m-1} (k ( 1 + \varepsilon ) + 3\Delta ) W (z_i, z_{i+1}) + 2 (k + \Delta) + 2(k +\Delta) \mathbb{I} (i=1) \\
        &\le \sum_{i=1}^{m} (k (1+\varepsilon) + 3\Delta) W (z_i, z_{i+1}) + 2(k+\Delta) \mathbb{I} (i=1).
    \end{align}
    Where in the last inequality we upper bound the middle term $2 (k + \Delta) \le (k (1 + \varepsilon) + 3\Delta) W(z_1,z_2) + (k (1 + \varepsilon) + 3\Delta) W(z_m,z_{m+1})$ for any arbitrary vertex $z_{m+1}$, since $W(\cdot,\cdot)$ is pointwise in $\{1,2\}$. This results in an upper bound on the total reward collected from all edges the optimizer plays which emanate from the vertices $z_1,z_2,\cdots,z_m$.
\end{proof}

We continue by bounding the rewards of actions/edges that emanate from vertices of a specific cycle of the graph $G'$; we show that the total reward collected by the optimizer for edges emerging from vertices in $\mathcal{C}_i$ is approximately upper bounded by the weight of edges in some $|\mathcal{C}_i|-1$ length path within $\mathcal{C}_i$ up to a factor of $k$. This follows using similar techniques used in the proof of the above lemma. We first apply \Cref{lemma:path} and bound the contribution of almost all edges emanating from all vertices of the cycle except one, and then prove that the edge emanating from the last vertex of the cycle cannot contribute too much.
\begin{lemma} \label{lemma:cycle}
Consider any \emph{cycle} $C_i$ in $G'$ composed of vertices $z_1,z_2,\cdots,z_m$. The total reward collected by the optimizer for edges which stem from vertices on $C_i$ is upper bounded by,
\begin{align}
    (k(1 + 3\varepsilon) + 7\Delta) \sum_{i \in [m] \setminus \{ i^\star \}} W (z_i, z_{i+1}) + 2(k + \Delta) \mathbb{I} (i=1) + \frac{1}{n}
\end{align}
for some $i^\star \in [m]$, where $z_{m+1} \triangleq z_1$.
\end{lemma}

\noindent We now combine the above results to get an upper bound on the total reward that heavy edges can contribute. This will involve combining the vertices in paths and cycles $C_i$ into a singular large Hamiltonian cycle, in a way such that the weight of this cycle continues to upper bound the total optimizer reward up to a multiplicative factor of $\approx k$. This results in an upper bound on the reward contribution of edges emanating from the heavy vertices. By \Cref{corr:excess} and \Cref{lemma:heavy}, this is the majority of the reward collected by this optimizer, giving the following upper bound on the optimizer reward.

\begin{lemma} \label{lemma:combine}
\label{lemma:alg-bound-MWU}
Assume $\Delta \ge \frac{1}{\eta} \log (2n^2 T)$. Then, the reward collected by the optimizer at the end of the repeated interaction is at most,
\begin{align}
    (k (1 + 5 \varepsilon) + 9 \Delta) \sum_{i = 1}^n W(z_i,z_{i+1}) + (2k + 2 \Delta + 4)
\end{align}
where $z_1 \to z_2 \to \cdots\to z_n \to z_1$ is some Hamiltonian cycle on the vertex set $V$, and $z_{n+1} = z_1$. This upper bound applies for any value of horizon $T$, as long as $\Delta$ is suitably large.
\end{lemma}

\noindent The proof of this result is provided in \Cref{app:combine}. The last step is to show that there exists an optimizer strategy which collects reward approximately lower bounded by the weight of the maximum weight Hamiltonian cycle (i.e., the solution to the maxTSP instance) up to a factor of $k$. This result is established below and proved in \Cref{app:opt-lb}.

\begin{lemma}[Bounding the reward of the best optimizer] \label{lemma:opt-lb} Consider any maximum weight Hamiltonian cycle in $G$, defined by the sequence of vertices $z_1^\star \to z_2^\star \to \cdots \to z_n^\star \to z_1^\star$. As long as $T \ge nk$, there exists an optimizer strategy such that the total reward collected is at least,
\begin{align}
    (1-\varepsilon)^2 k \times \left( 1 - \frac{2}{n} \right) \sum_{i=1}^n W ( z_i^\star, z_{i+1}^\star ),
\end{align}
where $z_{n+1}^\star = z_1^\star$, and assuming that $\eta \varepsilon k \ge \frac{1}{\varepsilon-\varepsilon^2} \log (n^3)$.
\end{lemma}

\noindent We now have all we need to complete the proof.

\paragraph{Proof sketch of \Cref{theorem:with_initialization}.}
Suppose we could approximate the reward of the best optimizer to a factor $\rho$ sufficiently close to $1$; by \Cref{lemma:alg-bound-MWU} we know that the reward collected by the optimizer is upper bounded by $k(1+\epsilon_1)$ times the weight of some Hamiltonian cycle in the maxTSP instance, where $\epsilon_1$ can be made arbitrarily close to $0$ as $T$ grows. On the other hand, we know by \Cref{lemma:opt-lb} that the best the optimizer can do is at least $k(1-\epsilon_2)$ times the length of the heaviest Hamiltonian cycle in the maxTSP instance, where $\epsilon_2$ can also be made arbitrarily close to $0$ as $T$ grows. By \Cref{theorem:maxTSP-hardness}, we know that unless $\Ptime = \NPtime$, it is impossible to approximate the weight of the optimal Hamiltonian cycle to within a $(1+c)$ factor for some $c > 0$. This will imply a lower bound on $\rho$ for polynomial time optimizers. Since the weight of the heaviest Hamiltonian cycle in the $(1,2)$-maxTSP instance is at least $n$, the best optimizer collects reward $\approx kn = T$. The multiplicative constant hardness guarantee thereby translates to an additive $\Omega(T)$ hardness guarantee for computing the reward of the best optimizer.

\paragraph{Full proof of \Cref{theorem:with_initialization}.}
Recall the parameterization $\eta = 1/T^{1-\alpha}$. With this, we choose,
\begin{align*}
    n &= T^{\min \{\alpha,\beta\}/3},\\
    k &= T^{1-\min \{\alpha,\beta\}/3},\\
    \Delta &= T^{1-\alpha}\log(2n^2T) 
\end{align*}
Notice that $nk = T$ and $k \gg \Delta$. We also have $\eta \varepsilon k = T^{\alpha-1} \cdot \varepsilon \cdot T^{1 - \min\{\alpha, \beta\}/3} \geq \varepsilon T^{2\alpha/3}$ and since $\varepsilon$ is a constant ,so we get $\eta \varepsilon k \geq \frac{1}{\varepsilon - \varepsilon^2}\log(n^3)$ for large enough $T$. Thus, the assumptions of \Cref{lemma:opt-lb,lemma:combine} are satisfied. Suppose the maximum reward that can be obtained by the optimizer is $R_{\text{opt}}$, and is upper bounded by \Cref{lemma:alg-bound-MWU},
\begin{align*}
    R_{\text{opt}} &\leq (k (1 + 5 \varepsilon) + 9 \Delta) \sum_{i=1}^n W(z_i,z_{i+1}) + (2k + 2 \Delta + 4) \\
    &\leq (k (1 + 5 \varepsilon + 2/n) + 11 \Delta + 4) \sum_{i=1}^n W(z_i,z_{i+1}).
\end{align*}
We know that the best optimizer achieves reward at least $R_{\text{lb}}$, which is given by \Cref{lemma:opt-lb},
\begin{equation*}
        R_{\text{lb}} \geq k(1-\varepsilon)^2 \left(1 - \frac{2}{n} \right) \sum_{i=1}^n W (z_i^\star,z_{i+1}^\star)
\end{equation*}
Suppose $R_\star$ is the optimal reward the optimizer can get in this instance. We will prove that if the optimizer can guarantee that $\frac{R_{\text{opt}}}{R_\star} \geq \rho$ for some $\rho > \frac{1069}{1070}$, then it can also guarantee a $\frac{1068}{1069}$ approximation for the maxTSP problem. If the optimizer runs in polynomial time, this contradicts \Cref{theorem:maxTSP-hardness}, which shows that unless $\Ptime = \NPtime$, there is no polynomial time algorithm for $(1,2)$-maxTSP achieving an approximation factor of $1067/1068 + \epsilon$ for any $\epsilon > 0$. Taking the ratio,
\begin{align}
    \rho \le \frac{R_{\text{opt}}}{R_\star} \le \frac{R_{\text{opt}}}{R_{\text{lb}}} &\le \frac{(k (1 + 5 \varepsilon + 2/n) + 11 \Delta + 4) \sum_{i=1}^n W(z_i,z_{i+1})}{k(1-\varepsilon)^2 \left(1 - \frac{2}{n} \right) \sum_{i=1}^n W (z_i^\star,z_{i+1}^\star)}  \\
    &\le \left(1 + C \left( \varepsilon + \frac{\Delta}{k} + \frac{1}{n} \right) \right)\frac{\sum_{i=1}^n W(z_i,z_{i+1})}{\sum_{i=1}^n W (z_i^\star,z_{i+1}^\star)}
\end{align}
Now, assuming that $\rho \ge 1069/1070$, adjusting $\varepsilon$ to be a sufficiently small constant, noting that $\Delta /k = \widetilde{\mathcal{O}} (T^{-2\alpha/3})$ and $1/n = T^{-\Omega (\beta)}$, for sufficiently large $T$ we have that,
\begin{align}
    \frac{\sum_{i=1}^n W(z_i,z_{i+1})}{\sum_{i=1}^n W (z_i^\star,z_{i+1}^\star)} \ge \frac{1068}{1069}
\end{align}
Which contradicts \Cref{theorem:maxTSP-hardness} if the optimizer runs in polynomial time. As a consequence, unless $\Ptime = \NPtime$, there is no polynomial time optimizer satisfying,
\begin{align}
    \frac{R_{\text{opt}}}{R_\star} \ge \rho \ge \frac{1069}{1070}.
\end{align}
Noting that $R_\star \ge R_{\text{lb}} \ge k(1-\varepsilon)^2 (1 - 2/n) \sum_{i=1}^n W(z_i^\star, z_{i+1}^\star) \ge (1-c_0 \varepsilon) T$ for some absolute constant $c_0 > 0$, the above is equivalent to the statement,
\begin{align}
    R_{\text{opt}} - R_\star \ge c_1 T
\end{align}
for some sufficiently small absolute constant $c_1 > 0$.

\subsection{A comment on the initialization}

Note that the proof in the previous section considers the specific initialization of the form \cref{init:1,init:2,init:3,init:4}. Consider a sequence of (non-distinct) edges $\Ei$ on the vertex set $V$ and let $\dinitin$ and $\dinitout$ denote the induced in-degrees and out-degrees. Consider a new initialization,
\begin{align}
    \text{For } \vin \in \Vin,\ r_0 (\vin) &= -k + 2 \dinitin (\vin) \label{init:1a}\\
    \text{For } \vout \in \Vout,\ r_0 (\vout) &= -k + 2 \dinitout (\vout) \label{init:2a}\\
    \text{For } e \in E \setminus \{ e^\dagger \},\ r_0 ( e ) &= \Ei ( e ) \label{init:3a}\\
    r_0 ( e^\dagger ) &= k \label{init:4a}
\end{align}
When $\Ei = \emptyset$, we go back to the one in \cref{init:1,init:2,init:3,init:4}. With such an initialization, the upper bound on the total reward of the optimizer strategy in \Cref{lemma:combine} still remains true, where $T$ is replaced by $T + |\Ei|$ everywhere. The idea is simple: the initialization in \cref{init:1a,init:2a,init:3a,init:4a} may be realized by starting out with the one in \cref{init:1,init:2,init:3,init:4}, and allowing the optimizer to first playing the edges in $\Ei$ in sequence (where no reward is collected), and then playing its original strategy. Thus, we may upper bound the reward collected by the optimizer by considering what it would have collected over the total horizon (including that collected when edges in $\Ei$ are played). This results in the following corollary of \Cref{lemma:alg-bound-MWU}.

\begin{corollary} \label{corr:edgeinit}
Suppose $\Delta \ge \frac{1}{\eta} \log (2n^2 (T+|\Ei|))$ The total reward collected by an optimizer strategy when MWU's initialization is of the form \cref{init:1a,init:2a,init:3a,init:4a} is upper bounded by,
\begin{align}
    (k (1 + 5 \varepsilon) + 9 \Delta) \sum_{i = 1}^n W(z_i,z_{i+1}) + (2k + 2 \Delta + 4)
\end{align}
for some Hamiltonian cycle $z_1 \to z_2 \to \cdots \to z_n \to z_1$ on $G$.
\end{corollary}

\section{Lower bound against MWU: removing the initialization}

In this section, we will show how to reduce the case of MWU without initialization to the case with initialization. To do this, we will modify the structure of the game discussed in \Cref{sec:structure} for optimizing against MWU with the initialization of the form \cref{init:1,init:2,init:3,init:4} and add a few additional actions. For $p \in \mathbb{N}$ to be determined later,
\begin{subequations}
\begin{align}
    &\textbf{Optimizer's action space } \mathcal{A} = \Ainit \cup \{ \blacklozenge \} \textbf{ of size } m = |V|^2 - |V| + 1 \\
    &\textbf{Learner's action space } \mathcal{B} = \Binit \cup \Binittilde \textbf{ of size } n = (p+1) (|V|^2 + |V|)
\end{align}
\end{subequations}
The optimizer's action space is augmented by a single special action $\blacklozenge$. On the other hand, the learner's action space is augmented by $\Binittilde$, which contains $p$ copies of each action in $\Binit$. Each of these $p$ copies will behave symmetrically with respect to the game matrices $\bm{A}$ and $\bm{B}$, and so we will simply refer to the ``counterpart'' of any $b \in \Binit$ as any one of its copies in $\Binittilde$. We will choose $p = T^{\beta/3}$, which will bring the size of the learner's action space to be $\mathcal{O} (T^\beta)$.
\medskip

\noindent The game matrices are amended as follows. For any learner edge $\widetilde{b} \in \Binittilde$ and any action $a \in \Ainit$,
\begin{align}
    \bm{A} ( a , \widetilde{b} ) = 0.
\end{align}
Furthermore, the optimizer collects no reward for playing $\blacklozenge$. Namely,
\begin{align} \label{eq:bl}
    \text{For all } b \in \mathcal{B},\ \bm{A} (\blacklozenge, b) = 0.
\end{align}
From the point of view of the learner, the actions in $\Binit$ and $\Binittilde$ are not symmetric. For a small constant $\beta > 0$ to be determined later, and any $b \in \Binit$,
\begin{align} \label{eq:blinit}
    &\bm{B} ( \blacklozenge, b ) = \begin{cases}
        k^\star/ T &\text{if } b = e^\dagger \\
        - k^\star/ T &\text{if } b \in \Vin \cup \Vout \\
        0 &\text{otherwise.}
    \end{cases},
\end{align}
Where $k^\star = \varepsilon (1-\varepsilon)^{-1} T^{1-\min \{ \alpha,\beta \}/3}$. For any $b \in \Binit$ and its counterpart $\widetilde{b} \in \Binittilde$,
\begin{align} \label{eq:gamma}
    &\bm{B} ( \blacklozenge, \widetilde{b} ) = \bm{B} ( \blacklozenge, b ) - \gamma,
\end{align}
where $\gamma = \beta/3 (1-\varepsilon)^{-1} T^{-\alpha} \log ( T )$. On the other hand, for any learner action $b \in \Binit$ and its set of counterparts $\widetilde{b} \in \Binittilde$, and any optimizer edge $a \in \Ainit$,
\begin{align}
    \bm{B} (a,\widetilde{b}) = \bm{B} (a,b)
\end{align}

\subsection{Proof of \Cref{theorem:noinit}}

The intuition behind the amended game structure is as follows. In the first iteration, observe that the presence of $p \gg 1$ copies of each action in $\Binit$ ensures MWU does not place a significant amount of probability mass on actions in $\Binit$. This in turn ensures that the optimizer cannot collect much reward in the first step. In order to mitigate this issue in future steps, notice that whenever the optimizer plays $\blacklozenge$, the reward on all actions in $\Binittilde$ decreases a little bit. Once $\blacklozenge$ is played sufficiently many times, the actions in $\Binittilde$ are downweighted enough to the point where MWU no longer places a significant amount of mass on them; the optimizer can begin to collect reward after this point. How many times should the optimizer play $\blacklozenge$ before this starts to happen?

\smallskip
It will turn out that the answer to this question will be $\approx (\eta \gamma)^{-1} \log (p)$ times. Furthermore, the choice of $\gamma$ is reverse engineered from the equation $(\eta \gamma)^{-1} \log (p) = (1-\varepsilon)^2 T$. Thus, a good optimizer would need to play the action $\blacklozenge$ $\approx (1-\varepsilon)^2 T$ times before the effect of the actions in $\Binittilde$ is nullified. This leaves out a small $\approx 2\varepsilon T$ portion of the horizon where the optimizer may collect rewards. This motivates the notion of the  ``critical time''.

\begin{definition}[Critical time]
Define $T_\blacklozenge^\star = \frac{1}{\eta \gamma}\log (p)$ as the \textit{critical time}. By choice of $\gamma$ and $p$, $T^\star_\blacklozenge = (1- \varepsilon)^2 T$. Define the smallest time $t$ where $|\mathcal{T}_\blacklozenge (t)| \ge (1-\varepsilon)^3 T$ as the \textit{lower critical time}.

\smallskip
\noindent If at any time $t$, $|\mathcal{T}_{\blacklozenge} (t)| < (1-\varepsilon)^3 T$, then $p \cdot \exp \left( - \eta \gamma |\mathcal{T}_{\blacklozenge} (t) | \right) \ge p^\varepsilon \gg 1$. Likewise, if at any time $t$, $|\mathcal{T}_{\blacklozenge} (t)| \ge (1-\varepsilon) T$, then $p \cdot \exp \left( - \eta \gamma |\mathcal{T}_{\blacklozenge} (t) | \right) \le p^{-\varepsilon} \ll 1$.
\end{definition}

The lower critical time is essentially a conservative notion of the critical time. At the critical time, the effect of the actions in $\Binittilde$ are comparable to those in $\Binit$. And prior to the lower critical time, the actions in $\Binittilde$ dominate all of those in $\Binit$ (in that MWU places very little mass on actions in $\Binit$). Next we define a version of MWU which merges the contribution to the reward where $\blacklozenge$ was played.

\begin{definition}[Reduced MWU] \label{def:reduced-MWU}
Consider some sequence of actions $\{ \bx(t) \}_{t \ge 1}$ played by the optimizer. Let $\mathcal{T}_{\blacklozenge} (t)$ denote the set of time-points until and including time $t$ where the optimizer played $\blacklozenge$ and $\mathcal{T}_{\cross \blacklozenge} (t)$ denote the set of time-points until and including time $t$ where the optimizer did not play $\blacklozenge$. Let $\yr (t)$ denote the ``reduced MWU'' where we assume that cumulative rewards of actions are computed across time points where $\blacklozenge$ is not played. Namely, using the definition,
\begin{align}
    \forall b \in \Binit,\ \rltilde (b,t) &= \frac{k^\star}{T} |\mathcal{T}_\blacklozenge(t)| \cdot \mathbb{I} (b = e^\dagger) + \sum_{t' \in \mathcal{T}_{\cross \blacklozenge} (t)} \bx (t')^\top \bm{B} \delta_b \\
    \forall b \in \Binittilde,\ \rltilde (\widetilde{b},t) &= \frac{k^\star}{T} |\mathcal{T}_\blacklozenge(t)| \cdot \mathbb{I} (\widetilde{b} = \widetilde{e}^\dagger) + \sum_{t' \in \mathcal{T}_{\cross \blacklozenge} (t)} \bx (t')^\top \bm{B} \delta_{\widetilde{b}}
\end{align}
\end{definition}

\begin{lemma} \label{lemma:reduced-MWU}
For any action $b \in \Binit$,
\begin{align}
    \Pr_{Y \sim \by(t)} (Y = b) =
        \frac{1}{1 + p \cdot \exp \left( - \eta \gamma |\mathcal{T}_{\blacklozenge} (t)| \right)} \cdot \Pr_{Y \sim \yr (t)} (Y = b)
\end{align}
\end{lemma}

\begin{lemma} \label{lemma:opt-p-theta}
Consider some sequence of actions $\{ \bx (t) \}_{t \ge 1}$ played by the optimizer. Let $\mathcal{T}_{\text{init}}^{1-\varepsilon} (t)$ denote the collection of timepoints, $\left\{ t' \in \mathcal{T}_{\cross \blacklozenge} (t) : |\mathcal{T}_{\blacklozenge} (t')| \ge (1 - \varepsilon)^3 T \right\}$: the set of time points where the optimizer did not play $\blacklozenge$, accrued after $\blacklozenge$ itself has been played at least $(1 - \varepsilon)^3 T$ times. Then, the cumulative reward of the optimizer is at most,
\begin{equation}
    \frac{2T}{p^\varepsilon} + \sum_{t \in \mathcal{T}_{\text{init}}^{1-\varepsilon} (t)} \bx(t)^\top \bm{A} \yr (t).
\end{equation}
\end{lemma}
\begin{proof}
The cumulative reward of the optimizer can be written as,
\begin{align}
    \sum_{t=1}^T \bx(t)^\top \bm{A} \by (t) = \sum_{t \in \mathcal{T}_{\cross \blacklozenge} (T)} \frac{1}{1 + p \cdot \exp (-\eta \gamma |\mathcal{T}_{\blacklozenge} (t)|)}\bx (t)^\top \bm{A} \yr(t)
\end{align}
where we use the fact that the optimizer collects no reward for picking $\blacklozenge$ or if the learner picks an action in $\Binittilde$, and on the remaining actions we use \Cref{lemma:reduced-MWU} to relate $\by(t)$ and the reduced MWU, $\yr (t)$. At any time $t$ prior to the lower critical time, we have that $p \cdot \exp \left( - \eta \gamma |\mathcal{T}_{\blacklozenge} (t) | \right) \ge p^\varepsilon \gg 1$. And therefore, by splitting the summation over $t \in \mathcal{T}_{\cross \blacklozenge} (T)$ into $\{ t : |\mathcal{T}_\blacklozenge (t)| < (1-\varepsilon)^3 T \}$ and the complement (which is nothing but $\mathcal{T}^{1-\varepsilon}_{\text{init}} (T)$), and noting that $\| \operatorname{Vec} (\bm{A}) \|_\infty \le 2$, we get the upper bound.
\end{proof}

Note that at any time $t$ where $\mathcal{T}_{\text{init}}^{1-\varepsilon} (t)$ is non-empty, $\blacklozenge$ has already been played $(1-\varepsilon)^3 T$ times. By the structure of the game, the action $e^\dagger$ has collected $\approx k^\star$ cumulative reward by this point, and likewise, actions in $\Vin \cup \Vout$ have cumulative reward $\approx - k^\star$. Observe the resemblance to the initialization in \cref{init:1,init:2,init:3,init:4}: we may now use techniques for bounding the reward of the optimizer for the case of MWU with initialization. In particular, the rewards collected by the optimizer are essentially bounded by \Cref{lemma:alg-bound-MWU}.

\begin{lemma} \label{lemma:opt-bound-MWU-noinit}
Assume that $\Delta \ge \frac{1}{\eta} \log (2n^2T)$. The cumulative reward collected by the optimizer when playing against MWU is upper bounded by,
\begin{align}
    2T^{1-\beta \varepsilon/3} + (k^\star (1 + 5 \varepsilon) + 9 \Delta) \sum_{i = 1}^n W(z_i,z_{i+1}) + (2k^\star + 2 \Delta + 4)
\end{align}
for some Hamiltonian cycle $z_1 \to z_2 \to \cdots \to z_n \to z_1$ on $G$.
\end{lemma}

The same argument will apply for lower bounding the total reward collected by the best optimizer strategy. We will consider the optimizer which plays $\blacklozenge$, $(1-\varepsilon) T$ times first, prior to following the optimizer strategy outlined in \Cref{lemma:opt-lb} with $k=k^\star$. Note that the remaining duration in the epoch is $\varepsilon T$, which is at least $n k^\star$.

\begin{lemma}
\label{lemma:opt-lb-noinit} Consider any maximum weight Hamiltonian cycle in $G$, defined by the sequence of vertices $z_1^\star \to z_2^\star \to \cdots \to z_n^\star \to z_1^\star$. There exists an optimizer strategy such that the total reward collected is at least,
\begin{align}
    R_{\text{lb}} \triangleq (1-\varepsilon)^3 k^\star \times \frac{1}{1 + p^{-\varepsilon}} \left( 1 - \frac{2}{n} \right) \sum_{i=1}^n W ( z_i^\star, z_{i+1}^\star ).
\end{align}
\end{lemma}

\noindent Finally, we may combine \Cref{lemma:opt-bound-MWU-noinit,lemma:opt-lb-noinit} to prove \Cref{theorem:noinit}, using the fact that both equations, up to an $\approx k^\star$ factor, capture the weight of a Hamiltonian cycle on $G$.

\subsection{Proof of \Cref{theorem:noinit}}

The proof of this result follows very similarly to that of \Cref{theorem:initialization}. The ratio of the optimizer's reward $R_{\text{opt}}$, to that of the best optimizer is at least,
\begin{align}
    \frac{R_{\text{opt}}}{R_\star} \le \frac{R_{\text{opt}}}{R_{\text{lb}}} &\le \frac{2T^{1-\beta \varepsilon/3} + (k^\star (1 + 5 \varepsilon) + 9 \Delta) \sum_{i = 1}^n W(z_i,z_{i+1}) + (2k + 2 \Delta + 4)}{(1 - \varepsilon)^3 k^\star \times \frac{1}{1 + p^{-\varepsilon}} \left( 1 - \frac{2}{n} \right) \sum_{i=1}^n W ( z_i^\star, z_{i+1}^\star )} \\
    &\le \big( 1 - C\varepsilon - o_T(1) \big) \cdot \frac{\sum_{i = 1}^n W(z_i,z_{i+1})}{\sum_{i=1}^n W ( z_i^\star, z_{i+1}^\star )}
\end{align}
for a sufficiently large absolute constant $C>0$ in the last inequality. Here we use the fact that $\Delta = \mathcal{O} (T^{1-\alpha} \log(T))$, while $k^\star = \Omega (T^{1-\alpha/3})$, and therefore $k^\star / \Delta = T^{\Omega(\alpha)}$, and the fact that $T^{1-\beta \varepsilon/3} \ll n k^\star = \varepsilon T$. \Cref{theorem:maxTSP-hardness} implies that unless $\Ptime = \NPtime$, a polynomial time algorithm returning a Hamiltonian cycle $z_1 \to z_2 \to \cdots \to z_n \to z_1$ for $(1,2)$-maxTSP must satisfy,
\begin{align}
     \frac{\sum_{i = 1}^n W(z_i,z_{i+1})}{\sum_{i=1}^n W ( z_i^\star, z_{i+1}^\star )} \le \frac{1067}{1068} + \epsilon
\end{align}
for any $\epsilon > 0$. Since $R_\star \ge R_{\text{lb}}$ and the latter is at least $(1-c_0\varepsilon) \varepsilon T - o(T)$ for some $c_0 > 0$ (\Cref{lemma:opt-lb-noinit}), and noting that $\varepsilon$ is a small constant, this implies that unless $\Ptime = \NPtime$, a polynomial time optimizer must satisfy,
\begin{align}
    R_\star - R_{\text{opt}} \ge c_1 \varepsilon T - o(T),
\end{align}
for some constant $c_1 > 0$. Since $\varepsilon$ is a small constant, this proves \Cref{theorem:noinit}.

\section{Conclusion and future questions}

We show that it is hard in general to optimize against a no-regret learner: specifically, against Hedge/MWU. This implies that positive results can only be obtained for games with specific structure, as was extensively studied in prior literature. It remains open what is the best runtime for the optimizer. While it is not polynomial in the game size unless $\Ptime = \NPtime$, can it be polynomial in $T$? Can the runtime depend exponentially only on $\min(|\mathcal{A}|,|\mathcal{B}|)$? Concretely, is there an algorithm that runs in time $\mathrm{poly}(|\mathcal{A}|,|\mathcal{B}|,T)e^{O(\min(|\mathcal{A}|,|\mathcal{B}|))}$? Further, beyond the special cases studied, is there a broader class of games in which it is possible to efficiently find optimize the optimizer's reward? Beyond MWU, there is an efficient optimization algorithm against no-swap regret learners \citep{deng2019strategizing}, but what about other learners? Taking the learner's perspective, can one analyze how well it performs against optimizers and which learning algorithms are desired in such a scenario? See related work for some studies in this direction.

\bibliographystyle{plainnat}
\bibliography{refs}

\appendix

\section{Optimizing against MWU with initialization}

In this section, we prove a lower bound on optimizing when the learner has non-zero reward history, formulated as follows:

\begin{theorem} \label{theorem:initialization}
    Fix absolute constants $\alpha, \beta \in (0,1]$. Suppose the learner plays MWU (Definition~\ref{def:MWU}) with a specific non-zero initialization of reward history, and learning rate parameterized as $\eta = 1/T^{1-\alpha}$. There exist game matrices $\bm{A}, \bm{B} \in [-1,1]^{|\mathcal{A}| \times |\mathcal{B}|}$, where $|\mathcal{A}|, |\mathcal{B}| \le T^\beta$, such that unless $\Ptime = \NPtime$, there exists no polynomial time optimizer which finds a sequence of pure strategies for the optimizer $\{ \bx (t) : t \in [T] \}$ satisfying,
\begin{align}
    R ( \{ \bx (t) : t \in [T] \}) \ge \max_{\{ \bx^\star (t) : t \in [T] \}} R ( \{ \bx^\star (t) : t \in [T] \}) - cT
\end{align}
For a sufficiently small absolute constant $c > 0$.
\end{theorem}

Furthermore, note that for any directed edge $(w,x) \in E$, the cumulative reward of the learner at time $t$ is,
\begin{align} \label{eq:cum-lr}
    \rl ((w,x),t) &= r_0 ((w,x)) + \sum_{t'=1}^t \bx(t')^\top \bm{B} \delta_{(w,x)} \\
    &= \begin{cases}
        E ((w,x),t) - \varepsilon \dout ( x, t ) \qquad &\text{if } (w,x) \ne e^\dagger, \\
        k - \varepsilon \dout ( x, t ) &\text{otherwise.}
    \end{cases}
\end{align}
where $\dout (x,t)$ is the out-degree of vertex $x$ at time $t$ (as induced by the edges played by the optimizer) and $E (e,t)$ counts the number of times the edge $e$ was played by the optimizer up to an including time $t$. \cref{eq:cum-lr} results from the fact that each time the edge $(w,x) \ne e^\dagger$ is played by the optimizer the cumulative reward increases by $1$, while each edge leaving $x$ removes a reward of $\varepsilon$. On the other hand, for any vertex $\vin \in \Vin$ or $\vout \in \Vout$,
\begin{align}
    \rl ( \vin, t) &= -k + 2 \din (\vin, t) \\
    \rl ( \vout, t) &= -k + 2 \dout (\vout, t)
\end{align}

\subsection{Proof of \Cref{lemma:excess}}

Recall that the reward of the actions of the learner at time $t$ is as follows:
\begin{align}
&\rl ((w,x),t) = \begin{cases}
    E ((w,x),t) - \varepsilon \dout ( x, t ) \qquad &\text{if } (w,x) \ne e^\dagger, \\
    k - \varepsilon \dout ( x, t ) &\text{otherwise.}
    \end{cases} \\
    &\rl ( \vin, t) = -k + 2 \din (v, t) \\
    &\rl ( \vout, t) = -k + 2 \dout (v, t)
\end{align}
Take $u \in V_{\text{excess}}(\Delta, t)$ for which $\max \{ \din (u,t), \dout (u,t) \}$ is maximized, and supposed without loss of generality that $\din(u,t) = \max \{ \din (u,t), \dout (u,t) \}$ (the proof still goes through if $\dout(u,t) = \max \{ \din (u,t), \dout (u,t) \}$ in an identical way). We claim that for each action $(w,x)$ of the learner in $E$ we have:
\begin{equation*}
    \rl(\uin, t) \geq \Delta + \rl((w,x), t)
\end{equation*}
Notice that 
\begin{align}
    \rl(\uin, t) &= -k + 2 \din(u, t) = -k + \din(u, t) +  \din(u, t)  \\
    & \ge -k + \din(u, t) + k + \Delta = \din(u, t) + \Delta
\end{align}
We also have for $(w,x) \neq e^\dagger$
\begin{align}
    \rl((w,x), t) = E ((w,x),t) - \varepsilon \dout ( x, t ) \leq E ((w,x),t) \leq \din(x, t) \leq \din(u,t)
\end{align}
and for $e^\dagger$
\begin{align}
    \rl(e^\dagger, t) < k < \din(u,t)
\end{align}
thus the claim is proven. Now let us look at the probability mass the learner will put on actions of $\Vin \cup \Vout$ in the next round, compared to the actions of the learner in $E$. Let us denote the latter probability with $p_E$. We have:
\begin{align}
    p_E = \frac{\sum_{(w,x) \in E} \exp(\eta \rl((w,x),t))}{\sum_{(w,x) \in E} \exp(\eta \rl((w,x),t)) + \sum_{v \in V} \exp(\eta\rl(\vin,t)) + \exp(\eta \rl(\vout,t))}
\end{align}
We have that the probability mass of the actions in $\Vin \cup \Vout$ is at least:
\begin{align}
    &\sum_{v \in V} \exp(\eta \rl(\vin,t)) +\exp(\eta\rl(\vout,t)) \\
    &\ge \exp(\eta\rl(\uin, t)) \\
    & \geq \exp(\eta \din(u,t)) \cdot \exp(\eta \Delta) = \exp(\eta \din(u,t)) \cdot \exp(\log(2 n^2 T)) \\
    &=  2 \exp(\eta \din(u,t)) \cdot n^2 T 
\end{align}
while the probability mass in the actions in $E$ are:
\begin{align}
    \sum_{(w,x) \in E} \exp(\eta \rl((w,x),t)) & \le \sum_{(w,x) \in E} \exp(\eta \din(u,t)) < n^2 \exp(\eta \din(u,t))
\end{align}
Thus we get that:
\begin{align}
    p_E \leq \frac{\exp(\eta \din(u,t)) \cdot n^2}{\exp(\eta \din(u,t))\cdot n^2 + 2\exp(\eta \din(u,t)) \cdot n^2 T} \leq \frac{1}{2 T+1} < \frac{1}{2 T}
\end{align}
Since the optimizer can  only earn reward from actions in $E$, the reward collected in the next round is upper bounded by $\frac{1}{T}$, as desired.

\subsection{Proof of \Cref{corr:excess}}
This follows directly from Lemma \ref{lemma:excess}. Since at time $\tmax$ we have that $V_{\text{excess}}$ is non-empty, that means that for rounds $\tmax + 1, \tmax + 2, \dots, T$ the optimizer can earn reward at most $\frac{1}{T}$. Summing up gives that the total reward that can be obtained by the optimizer:
\begin{equation*}
    R(\tmax) + \frac{T - \tmax}{T} < R(\tmax) + 1
\end{equation*}
as desired.

\subsection{Proof of \Cref{lemma:heavy}}

Prior to proving this result, we will introduce two auxiliary lemmas.

\begin{lemma} \label{lemma:low-reward}
At any time $t$, there will always exist an action for the learner having cumulative reward at least $k (1 - \varepsilon)$.
\end{lemma}

\paragraph{Proof of \Cref{lemma:low-reward}.}
Notice that initially, we have $\rl(e^\dagger, 0) = \rl((u^\dagger, v^\dagger), 0) =  k$. At some time $t$ we will have that $\rl(e^\dagger, t) = k - \varepsilon \din(v^\dagger, t)$. At time $t$ we also have for the action $\vin^\dagger$ that, $\rl(v^\dagger_{\text{in}}, t) = -k + 2\din(v^\dagger, t)$. Now consider $\max(\rl(e^\dagger, t), \rl(\vin^\dagger, t)$. This quantity is always greater than $k(1-\varepsilon)$, thus one of the two actions will always have reward greater than $k(1-\varepsilon)$, which proves the claim.

\begin{corollary}\label{corr:low-reward}
Consider any edge $(u,v) \ne (u^\dagger,v^\dagger)$ which was chosen fewer than $k ( 1 - \varepsilon ) - \Delta$ times until time $t$ by the optimizer. MWU places mass of at most $1/n^2 T$ on the corresponding learner edge $(u,v)$.
\end{corollary}

\paragraph{Proof of \Cref{corr:low-reward}.}
By lemma \ref{lemma:low-reward} we know that at any time $t$, there exists an action for the learner having cumulative reward at least $k(1-\varepsilon)$ and we know that action is not $(u,v)$. The probability mass MWU will then place on action $(u,v)$ will be at most,
\begin{equation}
    \frac{\exp(\eta(k(1-\varepsilon) - \Delta))}{\exp(\eta(k(1-\varepsilon) - \Delta)) + \exp(\eta k (1 - \varepsilon))} = \frac{1}{1 + \exp(\eta \Delta )} < \frac{1}{1 + 2n^2T} < \frac{1}{n^2T}
\end{equation}
as required.

\paragraph{Proof of \Cref{lemma:heavy}.} Consider actions/edges $(v,w)$ played by the optimizer, that contribute to his reward but do not contribute to $\rh(\Delta, t)$, i.e., where $v \not \in \Sh(\Delta, t)$. We will try to upper bound the contribution of these actions to the optimizer's rewards. If we take a look at the utility matrix of the optimizer, when the optimizer is plays such an action reward can be obtained in only two cases:
\begin{enumerate}
    \item The first case corresponds to when the learner plays an edge of the form $(u, v)$ for some $u \in V$. Since $v \not \in \Sh(\Delta, t)$, no edge that is incident to $v$ has been played enough to be considered heavy, i.e., $Q((u,v), 
    \tmax) \leq k(1-\varepsilon) - \Delta$. Since this action has not been played enough, it will be dominated by the action that has reward at least $k(1-\varepsilon)$ as explained by Lemma \ref{lemma:low-reward}, and therefore MWU will allocate very little probability mass on it and we can upper bound its contribution. Let us call the total reward the optimizer earns in this case, i.e., the reward the optimizer earns when playing edges of the form $(v,w)$ when the learner plays edges of the form $(u,v)$, be $R_1$.
    \item The second type of reward, results from when the learner also plays the same edge $(v,w)$. As the optimizer plays $(v,w)$ more and more, the equivalent action $(v,w)$ of the learner begins to accumulates reward. However, this reward will be relevant only when the edge $(v,w)$ has been played enough times from the optimizer that $(v,w)$ is no longer dominated by the action that has reward at least $k(1-\varepsilon)$ (cf. Lemma~\ref{lemma:low-reward}); this is the point when the combination $((v,w), (v,w))$ will start earning reward for the optimizer. However, note that the optimizer is constrained in playing $(v,w)$ a limited number of times, specifically less than $k + \Delta$ times since we are considering times $t \leq \tmax$, thus the reward in this case for the optimizer will also be limited. Let us denote the total reward the optimizer earns in this case, i.e. the reward the optimizer earns when playing edges of the form $(v,w)$ and the learner plays edges of the form $(v,w)$ as $R_2$.
\end{enumerate}
Let us now look at the reward the optimizer can earn in the above two cases. We make two claims; firstly, that $R_1 \leq 2$ and secondly that $R_2 \leq 2n(k\varepsilon + 2\Delta) + 1$. Putting the two claims together gives that:
\begin{equation}
    \ro (\tmax) \le \rh ( \Delta, \tmax ) + R_1 + R_2 \leq \rh ( \Delta, \tmax ) + 2n(k\varepsilon + 2\Delta) + 3
\end{equation}
as desired. It remains to prove the two claims. For the first claim, note that the actions of the form $(u,v)$ will have $\rl((u,v), t) \leq k(1-\varepsilon) - \Delta$, since by assumption, $v \not \in \Sh(\Delta, t)$. Using Corollary \ref{corr:low-reward} we can deduce that the mass assigned to these actions is upper bounded by $\frac{1}{n^2 T}$. There can be at most $n^2$ such choices of actions, and all these actions, across all time steps can earn reward at most,
\begin{equation}
    R_1 \leq 2 \cdot n^2 T \cdot \frac{1}{n^2 T} \leq 2
\end{equation}
which proves the first claim. For the second claim, take again an edge $(v,w)$ for the optimizer that is played more than $k(1-\varepsilon) - \Delta$. Note that for a fixed choice of $v$, there is only at most one edge $(v,w)$ that is played more than $k(1-\varepsilon) - \Delta$ times. If there were more than one, then $\dout(v, t)$ would exceed $k + \Delta$ contradicting that we are in time $t \leq \tmax$. We thus want to upper bound the reward earned by the optimizer when playing $(v,w)$. In the worst case the optimizer plays the action $(v,w)$, $k + \Delta$ times. Let us split the reward earned by that action in two time periods; the reward of the optimizer the first $k(1-\varepsilon) - \Delta$ the optimizer plays the action and the last $k\varepsilon + 2\Delta$ times. In the first interval, at times when the optimizer plays $(v,w)$, note that the learner has not built significant reward for this action $(v,w)$ action yet, specifically $\rl((v,w), t) \leq k(1-\varepsilon) - \Delta$. Since when playing $(v,w)$ the optimizer earns reward only when the learner is playing the same action. According again to Corollary \ref{corr:low-reward}, the learner will place mass at most $\frac{1}{n^2T}$ in this action, and therefore the optimizer can earn reward at most $\frac{1}{n^2T} \cdot 2 \cdot (k(1-\varepsilon) - \Delta) < \frac{1}{n^2} \cdot 2 \cdot T = \frac{2}{n^2}$. For the second part, the optimizer can earn reward at most $2$ for the remaining $\varepsilon k + 2\Delta$ rounds, thus making the total reward for the specific action $\frac{2}{n^2} + \varepsilon k + 2\Delta$. Summing up over all possible $v$'s we get that the total reward $R_2$ for the optimizer is:
\begin{equation}
    R_2 \leq n \cdot \left( \frac{2}{n^2} + \varepsilon k + 2\Delta \right) \leq 1 + n(\varepsilon k + 2 \Delta)
\end{equation}
which proves the second claim and therefore the lemma.

\subsection{Proof of \Cref{lemma:cycle}}

Consider the smallest time $t_i$ for each $i$ at which the vertices $z_i$ first appear in $\Sh(\Delta, t_i)$. As the optimizer plays edges, suppose the vertex $z_{i^*}$ is the last to appear last in $\Sh$ among all the $z_i$'s in this cycle. For all the edges in the graph $G'$ that belong in this cycle, we can bound their contribution in a similar way we did in Lemma \ref{lemma:path}, except the contribution coming from edges emanating from $z_{i^*}$. This contribution is:
\begin{equation}
    (k(1 + \varepsilon) + 3\Delta) \sum_{i \in [m] \setminus \{ i^\star\}} W (z_i, z_{i+1}) + 2(k + \Delta) \mathbb{I} (i=1)
\end{equation}
It ultimately remains to bound the contribution of the actions emanating from $i^*$. We will prove that the first $k(1-\varepsilon) - \Delta$ times the optimizer plays action $(z_{i^*}, z_{i^ + 1})$ very little reward is collected. Notice that the optimizer can earn reward either if the learner places probability mass on the same edge or on an edge that is of the form $(\cdot, z_{i^*})$. For the first case, the maximum reward the optimizer can obtain is:
\begin{equation}
    \sum_{j=0}^{k(1-\varepsilon) - \Delta}\frac{\exp(\eta j)}{\exp(\eta j) + \exp(\eta (k(1-\varepsilon) - \Delta))} \leq (k(1-\varepsilon) - \Delta) \cdot \exp(-\eta \Delta)\leq \frac{1}{n^2}
\end{equation}
For the second case, notice that each edge $(\cdot, z_{i^*})$ has been played at most $k(1-\varepsilon) - \Delta$ by the optimizer, since $i^*$ is the last node to become heavy, thus meaning that the learner will not place a lot of mass on these edges. This bounds the reward of the optimizer as follows:
\begin{equation}
    n \cdot 2 \cdot (k(1-\varepsilon) - \Delta) \cdot \frac{\exp(\eta (k(1-\varepsilon) - \Delta))}{\exp(\eta (k(1-\varepsilon) - \Delta)) + \exp(\eta (k(1-\varepsilon)))} < n \cdot 2 \cdot (k(1-\varepsilon) - \Delta) \cdot \exp(\eta \Delta) < \frac{1}{2n}
\end{equation}
For the remaining at most $\varepsilon k + 2\Delta$ times the optimizer plays an edge emanating from $z_{i^*}$ we may assume the maximum reward of $2$ is collected. The optimizer plays at most $\varepsilon k + 2\Delta$ such edges, and thereby earns reward at most $2(\varepsilon k + 2\Delta)$ in this case. This bounds the reward emanating from $z_{i^*}$ to:
\begin{equation}
    \frac{1}{n^2} + \frac{1}{2n} + 2(\varepsilon k + 2\Delta) < \frac{1}{n} + 2(\varepsilon k + 2\Delta)
\end{equation}
and thus the total reward that can be earned from the cycle $C_i$ is:
\begin{align}
    &(k(1 + \varepsilon) + 3\Delta) \sum_{i \in [m] \setminus \{ i^\star\}} W (z_i, z_{i+1}) + 2(k + \Delta) \mathbb{I} (i=1)
    + \frac{1}{n} + 2(\varepsilon k + 2\Delta) < \\
    &(k(1 + 3\varepsilon) + 7\Delta) \sum_{i \in [m] \setminus \{ i^\star\}} W (z_i, z_{i+1}) + 2(k + \Delta) \mathbb{I} (i=1) + \frac{1}{n}
\end{align}
as required.

\subsection{Proof of \Cref{lemma:combine}} \label{app:combine}

We will show that the total value collected by the optimizer is at most,
\begin{equation}
    (1) + \left(2n(k\varepsilon + 2\Delta) + 3\right) + \left((k (1 + 3 \varepsilon) + 7 \Delta) \sum_{i = 1}^n W(z_i,z_{i+1}) +  2(k+\Delta) + 1\right)  \label{eq:upper-bound-reward}
\end{equation}
which is upper bounded by the quantity in the statement of the lemma. Let us recall the different ways the optimizer collects rewards.
\begin{enumerate}
    \item If time $\tmax$ is reached, by \Cref{corr:excess} we argued that then onward, the optimizer will collect reward at most $1$. This accounts for the first term in \cref{eq:upper-bound-reward}.
    \item From time $0$ up until time $\tmax$, we argued through \Cref{lemma:heavy} that only at most $2n(k\varepsilon + 2\Delta)$ reward can come from edges whose vertices do not emanated from vertices in $\Sh(\Delta, t)$. This accounts for the second term in Eq. \ref{eq:upper-bound-reward}.
    \item Finally, we use Lemmas \ref{lemma:path} and \ref{lemma:cycle} to argue that the optimizer can make at most $(k (1 + 3 \varepsilon) + 7 \Delta) \sum_{i = 1}^n W(z_i,z_{i+1}) +  2(k+\Delta) + 1$ where $z_1 \to z_2 \to \cdots\to z_n \to z_1$ is some Hamiltonian cycle. Suppose the optimizer's play induces the graph $G$' as defined previously. We will prove that the paths and cycles formed can be combined in order to create a Hamiltonian cycle, the weight of which will upper bound (up to some multiplicative and additive constants) the reward that can be obtained in this scenario by the optimizer. Suppose $C_1, C_2, \dots, C_l$ are all paths and $C_{l+1}, C_{l+2}, \dots, C_t$ are cycles, and $C_1$ contains the special edge $e^\dagger$ - we will later argue how to prove the claim even if the connected component of the special edge is a cycle. We first combine all the paths into one. Suppose path $C_i$ is $z^i_1 \to z^i_2 \to \cdots\to z^i_{m_i}$. We will combine the paths by attaching the last vertex of path $C_i$ to the first vertex of path $C_{i+1}$, i.e. $z^i_{m_i} \to z^{i+1}_{1}$. The combination of all paths will then be:
    \[
        z^1_1 \to z^1_2 \to \cdots\to z^1_{m_1} \to z^2_1 \to z^2_2 \to \dots z^2_{m_2} \to \dots \to z^l_1 \to \dots \to z^l_{m_l}
    \]
    If we sum up the rewards of the paths independently and use Lemma \ref{lemma:path}, the paths will earn rewards at most
    \[
        (k(1 + \varepsilon) + 3\Delta) \sum_{i=1}^l \sum_{j=1}^{m_i} W(z^i_{j}, z^i_{j+1}) + 2(k+\Delta)
    \]
    while the newly constructed path will earn reward
    \[
        (k(1 + \varepsilon) + 3\Delta) \sum_{i=1}^l \left(\left(\sum_{j=1}^{m_i-1} W(z^i_{j}, z^i_{j+1}) \right) + \mathbb{I}[i \neq l]W(z^i_{m_i}, z^{i+1}_{1})\right)+ 2(k+\Delta)
    \]
    which is a clear upper bound. Now we continue by combining the cycles. Suppose for cycle $C_i$ where $t\geq i > l$ and the cycle is $z^i_1 \to z^i_2 \to \cdots\to z^i_{m_i} \to z^i_{1} $ the vertex $i^*$ as described in Lemma \ref{lemma:cycle} is denoted by the vertex $z^i_{m_i}$. We will combine the cycles as follows, to create a big cycle:
    \[
        z^l_1 \to z^l_2 \to \cdots\to z^l_{m_{l}} \to z^{l+1}_{1} \to \dots z^{l+1}_{m_{l+1}} \to \dots \to z^t_{m_t} \to z^l_1
    \]
    We essentially delete the last edge of each cycle and redirect the last vertex to the first vertex of the next cycle. Note that the reward of this big cycle according to Lemma \ref{lemma:cycle} is:
    \begin{equation}
        (k(1 + 3\varepsilon) + 7\Delta) \sum_{i=l+1}^t \left(\left(\sum_{j=1}^{m_i-1} W(z^i_{j}, z^i_{j+1}) \right) + \mathbb{I}[i \neq t]W(z^i_{m_i}, z^{i+1}_{1})\right) 
    \end{equation}
    which is upper bounding the rewards -when adding a constant $1$ to cover for $\frac{t}{n}$ -  of the individual cycles, which is:
    \begin{equation}
        (k(1 + 3\varepsilon) + 7\Delta) \sum_{i=l+1}^t \sum_{j=1}^{m_i-1} W(z^i_{j}, z^i_{j+1}) + \frac{t}{n}
    \end{equation}
    Lastly, now we connect the big path and the big cycle together as follows: direct $z^l_{m_l} \to z^{l+1}_1$ and redirect $z^t_{m_t} \to z^1_1$. This will close the cycle and upper bound the rewards obtained by the paths and cycles together. The reward of the Big cycle will then be:
    \[
        (k (1 + 3 \varepsilon) + 7 \Delta) \sum_{i = 1}^n W(z_i,z_{i+1}) +  2(k+\Delta) + 1
    \]
    as we wanted. In the case where $C_1$ was not a path but a cycle, we would have begun with the cycles instead of the paths and do the same thing with the only difference being that we would move over the $2(k+\Delta)$ factor to the cycles instead of the paths.
\end{enumerate}

\subsection{Proof of \Cref{lemma:opt-lb}} \label{app:opt-lb}

Without loss of generality assume $z_1^\star = v^\dagger$ and we define $z_{-1}^\star = u^\dagger$ and $z_{n+1}^\star = z_1^\star$. Define $Z$ as the set of edges $\{ (z_{i-1}^\star,z_{i}^\star) : i \in [n+1] \}$. Consider the sequence of optimizer pure strategies which plays the edge $(z_i^\star,z_{i+1}^\star)$ for $k(1- \varepsilon + \varepsilon^2)$ rounds for each $i$, and repeats across $i \in [n]$ (assuming $z_{n+1}^\star = z_1^\star$). For the last $nk(\varepsilon-\varepsilon^2)$ we lower bound the rewards of the optimizer by $0$ - basically disregarding those rounds. The total reward collected by this sequence of pure strategies is lower bounded next. At initialization, MWU associates a cumulative reward of $k$ with the action $(u^\dagger,v^\dagger)$ and $\le 0$ for the remaining actions. We argue that in each epoch $i \le n-1$, the total mass MWU places on the actions $(z_{i-1}^\star,z_i^\star)$ and $(z_i^\star,z_{i+1}^\star)$ is close to $1$ in most of the iterations within this epoch.

At the end of the $i^{\text{th}}$ epoch (i.e., $t=k(1-\varepsilon + \varepsilon^2) \times i + 1$), the cumulative rewards on the actions are,
\begin{align} \label{eq:cumr-lr-exact}
    B_0(b) + \sum_{s=1}^t B ((u_s,v_s), b) = \begin{cases}
        k (1 - \varepsilon) (1 -  \varepsilon+ \varepsilon^2) &\text{if } b = (z_j^\star, z_{j+1}^\star) \text{ for } j \le i-1,\\
        k (1 -   \varepsilon+ \varepsilon^2) \qquad &\text{if } b = (z_i^\star, z_{i+1}^\star), \\
        0 &\text{if } b = (z_j^\star, z_{j+1}^\star) \text{ for } j \ge i+1,\\
        0 &\text{if } b \in E \setminus Z,\\
        k(1 -2  \varepsilon+ 2 \varepsilon^2) &\text{if } b = (z_j^\star)' \text{ for } j \le i\\
        -k &\text{if } b = (z_j^\star)' \text{ for } j \ge i+1\\
        k(1-2  \varepsilon+2  \varepsilon^2) &\text{if } b = (z_j^\star)'' \text{ for } j \le i+1\\
        -k &\text{if } b = (z_j^\star)'' \text{ for } j \ge i+2
    \end{cases}
\end{align}
Here, for $v\in V$, we denote $(v)'$ as its corresponding vertex in $\Vout$, and $(v)''$ as its corresponding vertex in $\Vin$. In the $(i+1)^{\text{th}}$ epoch, the edge $(z_{i+1}^\star,z_{i+2}^\star)$ is chosen a number of times. Having chosen it $p$ times thus far, the cumulative rewards of actions are updated as,
\begin{align}
    B_0(b) + \sum_{s=1}^t B ((u_s,v_s), b) = \begin{cases}
        k (1 - \varepsilon) (1 -  \varepsilon + \varepsilon^2) &\text{if } b = (z_j^\star, z_{j+1}^\star) \text{ for } j \le i-1,\\
        k (1 -\varepsilon +   \varepsilon^2) - \varepsilon p \qquad &\text{if } b = (z_i^\star, z_{i+1}^\star), \\
        p \qquad &\text{if } b = (z_{i+1}^\star, z_{i+2}^\star), \\
        0 &\text{if } b = (z_j^\star, z_{j+1}^\star) \text{ for } j \ge i+2,\\
        0 &\text{if } b \in E \setminus Z,\\
        k(1-2  \varepsilon+2  \varepsilon^2) &\text{if } b = (z_j^\star)' \text{ for } j \le i\\
        -k + 2p &\text{if } b = (z_{i+1}^\star)' \\
        -k &\text{if } b = (z_j^\star)' \text{ for } j \ge i+2\\
        k(1-2  \varepsilon +  2  \varepsilon^2) &\text{if } b = (z_j^\star)'' \text{ for } j \le i\\
        -k + 2p &\text{if } b = (z_{i+2}^\star)''\\
        -k &\text{if } b = (z_j^\star)'' \text{ for } j \ge i+3
    \end{cases} \label{eq:cum-reward-opt-intermediate}
\end{align}
It is easy to verify plugging in $p = k(1-  \varepsilon +\varepsilon^2)$, we arrive at the cumulative rewards at the end of the $(i+1)^{\text{th}}$ epoch, which matches what is observed plugging in $i=i+1$ in eq.~\eqref{eq:cumr-lr-exact}.

In the $(i+1)^{\text{th}}$ epoch, for any $p \le k(1-\varepsilon+\varepsilon^2)$, observe that the total mass MWU associates with the edges $(z_{i}^\star,z_{i+1}^\star)$ and $(z_{i+1}^\star,z_{i+2}^\star)$ is lower bounded by,
\begin{align} \label{eq:mwu-lb}
    \frac{e^{\eta (k ( 1 - \varepsilon+\varepsilon^2) - \varepsilon p)} + e^{\eta p}}{e^{\eta (k ( 1 - \varepsilon+ \varepsilon^2) - \varepsilon p)} + e^{\eta p} + (|\Binit|-2) e^{\eta k (1 - 2\varepsilon+2  \varepsilon^2)}}
\end{align}
This uses the fact that for the learner, the maximum cumulative reward on any action which is not one of $(z_{i}^\star,z_{i+1}^\star)$ or $(z_{i+1}^\star,z_{i+2}^\star)$ is upper bounded by $\max \{ k(1-\varepsilon)(1- \varepsilon+\varepsilon^2), k (1 -2\varepsilon+ 2  \varepsilon^2) \} \le k(1 - 2\varepsilon+ 2  \varepsilon^2)$ assuming that $\varepsilon$ is sufficiently small and $\varepsilon$ is bounded away from $1$. For any value of $p \le (1-\varepsilon)^2k$, using the fact that $|\Binit| \le 2 \binom{n}{2} + 2n = n^2 + n$, eq.~\eqref{eq:mwu-lb} is further lower bounded by,
\begin{align*}
     &\frac{e^{\eta (k ( 1 - \varepsilon +\varepsilon^2) - \varepsilon p)} + e^{\eta p}}{e^{\eta( k ( 1 - \varepsilon+  \varepsilon^2) - \varepsilon p)} + e^{\eta p} + (|\Binit|-2) e^{\eta k (1 - 2\varepsilon+ 2  \varepsilon^2)}} =  \left(\frac{e^{\eta( k ( 1 - \varepsilon+  \varepsilon^2) - \varepsilon p)} + e^{\eta p} + (|\Binit|-2) e^{\eta k (1 - 2\varepsilon+ 2  \varepsilon^2)}}{e^{\eta (k ( 1 - \varepsilon +\varepsilon^2) - \varepsilon p)} + e^{\eta p}} \right)^{-1} \\
     & \geq \left( 1 + \frac{(n^2 + n) e^{\eta k ( 1 - 2\varepsilon + 2  \varepsilon^2)}}{e^{\eta (k ( 1 - \varepsilon +\varepsilon^2) - \varepsilon p)} + e^{\eta p}}\right)^{-1} \ge 1 - \frac{(n^2 + n) e^{\eta k ( 1 - 2\varepsilon + 2  \varepsilon^2)}}{e^{\eta (k ( 1 - \varepsilon +\varepsilon^2) - \varepsilon p)} + e^{\eta p}}  \ge 1 - \frac{(n^2 + n) e^{\eta k ( 1 -2\varepsilon+ 2  \varepsilon^2)}}{e^{\eta (k ( 1 - \varepsilon+  \varepsilon^2) - \varepsilon p )} }  \\
    & \geq 1 - (n^2 + n) e^{\eta (k (-\varepsilon +  \varepsilon^2) + \varepsilon p)} \geq 1 - (n^2 + n) e^{\eta (k (-\varepsilon +  \varepsilon^2) + \varepsilon (1-\varepsilon)^2 k)} = 1 - (n^2 + n) e^{\eta k(-\varepsilon +  \varepsilon^2 + \varepsilon - 2\varepsilon^2 + \varepsilon^3)}\\
    &= 1 - (n^2 + n) e^{- \eta (\varepsilon - \varepsilon^2) \varepsilon k}. \label{eq:66}
\end{align*}
Note that $(\varepsilon-\varepsilon^2) \eta \varepsilon k \ge \log (n^3)$, by assumption, thus the total probability MWU associates with the edges $(z_i^\star,z_{i+1}^\star)$ and $(z_{i+1}^\star, z_{i+2}^\star)$ in the first $\Delta_1 \triangleq (1-\varepsilon^2) k$ steps of the $(i+1)^{\text{th}}$ epoch is at least $1 - 2/n$.

Note that in the $(i+1)^{\text{th}}$ epoch, the optimizer always plays the edge $(z_{i+1}^\star,z_{i+2}^\star)$. This edge collects reward for the optimizer based on the probability mass MWU associates with edges of the form $(\cdot,z_{i+1}^\star)$ and with the edge $(z_{i+1}^\star,z_{i+2}^\star)$. MWU places at least $1-2/n$ mass on these kinds of edges in the first $\Delta_1$ iterations of this epoch. Therefore, in the $(i+1)^{\text{th}}$ epoch, the optimizer collects reward of at least,
\begin{align}
    \Delta_1 \left( 1 - \frac{2}{n} \right) W(z_{i+1}^\star, z_{i+2}^\star)
\end{align}
Summing this over all values of $i = 0,1,\cdots,n-1$, the total reward collected by the optimizer is at least,
\begin{align}
    \Delta_1 \left( 1 - \frac{2}{n} \right) \sum_{i=1}^n W(z_{i}^\star, z_{i+1}^\star)
\end{align}
Up to a multiplicative factor of $\Delta_1 (1 - 2/n) =(1-\varepsilon)^2 k (1 - \frac{2}{n})\approx k$, the reward collected by this sequence of optimizer pure strategies matches the weight of the heaviest Hamiltonian cycle in $G$.

\section{Optimizing against MWU without initialization}

\subsection{Proof of \Cref{lemma:reduced-MWU}}

Note that the cumulative rewards $\rl (b,t)$ of MWU can be written as,
\begin{align} \label{eq:38}
    \forall b \in \Binit,\ \rl (b,t) &= \rltilde (b,t) + \sum_{t' \in \mathcal{T}_{\blacklozenge} (t)} \bx (t')^\top \bm{B} \delta_b - \frac{k^\star}{T} |\mathcal{T}_\blacklozenge(t)| \cdot \mathbb{I} (b = e^\dagger).
\end{align}
For every action $b \in \Binit \setminus \{ e^\dagger \}$, the second term is $0$ by the structure of the game (\cref{eq:blinit}), as is the third term. For the action $b = e^\dagger$, the second and third terms cancel each out, by the structure of the game (also by \cref{eq:blinit}). Therefore,
\begin{align} \label{eq:39}
    \forall b \in \Binit,\ \rl (b,t) = \rltilde (b,t)
\end{align}
On the other hand, for actions in $\Binittilde$, we have,
\begin{align}
    \forall \widetilde{b} \in \Binittilde,\ \rl (\widetilde{b},t) &= \rltilde (\widetilde{b},t) + \sum_{t' \in \mathcal{T}_{\blacklozenge} (t)} \bx (t')^\top \bm{B} \delta_{\widetilde{b}} - \frac{k^\star}{T} |\mathcal{T}_\blacklozenge(t)| \cdot \mathbb{I} (\widetilde{b} = \widetilde{e}^\dagger) \\
    &\overset{(i)}{=} \rltilde (\widetilde{b},t) + \sum_{t' \in \mathcal{T}_{\blacklozenge} (t)} \bx (t')^\top \bm{B} \delta_{b} - \gamma |\mathcal{T}_\blacklozenge (t)| - \frac{k^\star}{T} |\mathcal{T}_\blacklozenge(t)| \cdot \mathbb{I} (b = e^\dagger) \\
    &\overset{(ii)}{=} \rltilde (\widetilde{b},t) - \gamma |\mathcal{T}_\blacklozenge (t)|
\end{align}
where in $(i)$ we use $b$ to denote the counterpart of $\widetilde{b}$ in $\Binit$, and use the structure of the game (\cref{eq:gamma}). In $(ii)$ we use the same argument in going from \cref{eq:38} to \cref{eq:39}. 
Therefore for any action $b \in \Binit$,
\begin{align}
    \Pr_{Y \sim \by(t)} (Y = b) &= \frac{\exp \left( \eta \rl (b,t) \right)}{\sum_{b' \in \Binit} \exp \left( \eta \rl (b',t) \right) + \sum_{b' \in \Binittilde } \exp \left( \eta \rl (b',t) \right)} \\
    &= \frac{\exp \left( \eta \rltilde (b,t) \right)}{\sum_{b' \in \Binit} \exp \left( \eta \rltilde (b',t) \right) + p \sum_{b' \in \Binit} \exp \left( \eta \rltilde (b',t) \right) \cdot \exp \left( -\eta \gamma |\mathcal{T}_{\blacklozenge} (t)| \right) } \\
    &= \frac{1}{1 + p \cdot \exp (- \eta \gamma |\mathcal{T}_{\blacklozenge} (t)|)} \cdot \frac{\exp \left( \eta \rltilde (b,t) \right)}{\sum_{b' \in \Binit} \exp \left( \eta \rltilde (b',t) \right)}
\end{align}
Using the definition of reduced MWU completes the proof.

\subsection{Proof of \Cref{lemma:opt-bound-MWU-noinit}}

This is the cumulative reward collected by the optimizer when playing against MWU over a reduced game which ignores the steps where the optimizer plays $\blacklozenge$, accrued prior to the lower critical time. Prior to the lower critical time, the optimizer may have played actions which were not $\blacklozenge$. Thus, at the lower critical time, the reduced MWU learner essentially begins with an ``initialization'' (in the same sense as \cref{init:1a,init:2a,init:3a,init:4a}) corresponding to whichever non-$\blacklozenge$ optimizer edges were played. Formally, at time any time $t$ after the lower critical time, $|\mathcal{T}_\blacklozenge (t)| \ge (1 - \varepsilon) T_\blacklozenge^\star$, and any $b \in \Binit$,
\begin{align}
    \rltilde (b,t) &= \frac{k^\star}{T} |\mathcal{T}_\blacklozenge (t)| \cdot \mathbb{I} (b=e^\dagger) + \sum_{t' \in \mathcal{T}_{\cross \blacklozenge} (t)} \bx (t')^\top \bm{B} \delta_b \\
    &= \frac{k^\star}{T} |\mathcal{T}_\blacklozenge (t)| \cdot \mathbb{I} (b=e^\dagger) + \sum_{t' \in \mathcal{T}_{\cross \blacklozenge} (t) \setminus \mathcal{T}_{\text{init}}^{1-\varepsilon} (t)} \bx (t')^\top \bm{B} \delta_b + \sum_{t' \in \mathcal{T}_{\text{init}}^{1-\varepsilon} (t)} \bx (t')^\top \bm{B} \delta_b \\
    &= r_0 (b) + \sum_{t' \in \mathcal{T}_{\text{init}}^{1-\varepsilon} (t)} \bx (t')^\top \bm{B} \delta_b.
\end{align}
Where, $r_0 (b)$ is an initialization reward precisely of the form \cref{init:1a,init:2a,init:3a,init:4a}, where the induced value of $k$ is $\frac{k^\star}{T} |\mathcal{T}_\blacklozenge (t)|$: this is not difficult to see, since $\sum_{t' \in \mathcal{T}_{\cross \blacklozenge} (t) \setminus \mathcal{T}_{\text{init}}^{1-\varepsilon} (t)} \bx (t')^\top \bm{B} \delta_b$ is an initial reward obtained by playing some set of optimizer edges $\Ei$ (corresponding to those played at time $\mathcal{T}_{\cross \blacklozenge} (t) \setminus \mathcal{T}_{\text{init}}^{1-\varepsilon} (t)$: the non-$\blacklozenge$ actions played prior to the lower critical time). While the induced value of $k$ depends on the sequence of actions played thus far, the fact that $(1-\varepsilon)^3 T \le |\mathcal{T}_\blacklozenge (t)| \le T$ constrains it very tightly. The former induces $k = (1-\varepsilon)^3 T \times k^\star / T = (1-\varepsilon)^2 k^\star$ and the latter induces $k = T \times k^\star/ T = k^\star$.

With this discussion, we may bound the value collected by the optimizer using \Cref{corr:edgeinit}. Noting the fact that the optimizer has played $\blacklozenge$ at least $(1-\varepsilon)^3 T$ times, the remaining actions must have been played at most $T-(1-\varepsilon)^3 T$ times. Note that the reduced MWU is only determined at the time points the optimizer does not play $\blacklozenge$, which is a very small fraction of the overall horizon, covering at most $\le T - (1-\varepsilon)^3 T \le 3 \varepsilon T$ steps. In particular, for the induced initialization at the lower critical time, the set of edges in the initialization, $\Ei$, is $\mathcal{T}_{\cross \blacklozenge} (T) \setminus \mathcal{T}^{1-\varepsilon}_{\text{init}} (T)$, and the remaining duration of the horizon, which we denote $\Ti$, is at most $|\mathcal{T}^{1-\varepsilon}_{\text{init}} (T)|$. The overall effective horizon (cf. \Cref{corr:edgeinit}), $|\Ei| + \Ti$ is upper bounded by $T - (1-\varepsilon)^3 T \le 3\varepsilon T$ assuming that $\mathcal{T}^{1-\varepsilon}_{\text{init}} (T)$ is non-empty. In particular, as a consequence of \Cref{corr:edgeinit}, with the choice $k = k^\star$, results in the upper bound: for any $\Delta \ge \frac{1}{\eta} \log (2n^2 T)  \ge \frac{1}{\eta} \log (2n^2 (1-\varepsilon)T) = \frac{1}{\eta} \log (2n^2 (\Ti+|\Ei|))$,
\begin{align}
    \sum_{t \in \mathcal{T}^{1-\varepsilon}_{\text{init}} (T)} \bx (t)^\top \bm{A} \yr (t) \le (k^\star (1 + 5 \varepsilon) + 9 \Delta) \sum_{i = 1}^n W(z_i,z_{i+1}) + (2k^\star + 2 \Delta + 4)
\end{align}
The proof concludes by combining with \Cref{lemma:opt-p-theta} and the choice $p = T^{\beta/3}$.

\subsection{Proof of \Cref{lemma:opt-lb-noinit}}

Consider the optimizer strategy which plays the action $\blacklozenge$, precisely $(1-\varepsilon)T$ times. On the remaining iterations, which is of length $\varepsilon T$, the optimizer follows the strategy outlined in \Cref{lemma:opt-lb}, on a horizon of length $\varepsilon T$. Firstly, observe that, by \Cref{lemma:reduced-MWU}, at any time $t > (1-\varepsilon)T$,
\begin{align}
    \Pr_{Y \sim \by(t)} (Y = b) &= \frac{1}{1 + p \cdot \exp \left( - \eta \gamma |\mathcal{T}_{\blacklozenge} (t)| \right)} \cdot \Pr_{Y \sim \yr (t)} (Y = b) \\
    &= \frac{1}{1 + p^{-\varepsilon}} \cdot \Pr_{Y \sim \yr (t)} (Y = b)
\end{align}
where in the last equation, we use the choice of $\gamma$ and $p$. At time $t = (1-\varepsilon)T + 1$, the cumulative reward on actions is of the form,
\begin{align}
    \text{For } \vin \in \Vin,\ \rl (\vin,t) &= - (1-\varepsilon) k^\star = \varepsilon T^{1-\min \{\alpha,\beta\}/3} \\
    \text{For } \vout \in \Vout,\ \rl (\vout,t) &= - \varepsilon T^{1-\min \{\alpha,\beta\}/3} \\
    \text{For } (u,v) \in E \setminus \{ e^\dagger \},\ \rl ( (u,v),t ) &= 0 \\
    \rl ( e^\dagger,t ) &= \varepsilon T^{1-\min \{\alpha,\beta\}/3}
\end{align}
Note that this optimizer has $\Ti = \varepsilon T$ time remaining in the horizon, and the induced value of $k = \varepsilon T^{1-\min \{\alpha,\beta\}/3}$: in particular, since, $n \times k = \varepsilon T \triangleq \Ti$, and therefore the horizon is sufficiently long to obtain the reward of the policy described in \Cref{lemma:opt-lb}.

The cumulative reward of the optimizer can be written as,
\begin{align}
    \sum_{t=1}^T \bx(t)^\top \bm{A} \by (t) &= \frac{1}{1 + p^{-\varepsilon}} \sum_{t = (1-\varepsilon)T+1}^T \bx (t)^\top \bm{A} \yr(t) \\
    &\ge (1 - \varepsilon)^2 \varepsilon T^{1-\min \{ \alpha,\beta\}/3} \times \frac{1}{1 + p^{-\varepsilon}} \left( 1 - \frac{2}{n} \right) \sum_{i=1}^n W ( z_i^\star, z_{i+1}^\star ) \\
    &\ge (1 - \varepsilon)^3 k^\star \times \frac{1}{1 + p^{-\varepsilon}} \left( 1 - \frac{2}{n} \right) \sum_{i=1}^n W ( z_i^\star, z_{i+1}^\star )
\end{align}
which completes the proof.

\end{document}